\documentclass[11pt,letterpaper]{article}
\usepackage{macros}

\usepackage{url}

\makeatletter
\def\url@leostyle{%
  \@ifundefined{selectfont}{\def\UrlFont{\sf}}{\def\UrlFont{\footnotesize\ttfamily}}}
\makeatother
\begin{document}
\title{New SDP Roundings and Certifiable Approximation for Cubic Optimization}

 \author{
 Jun-Ting Hsieh\thanks{Carnegie Mellon University, \texttt{juntingh@cs.cmu.edu}. Supported in part by NSF CAREER Award \#2047933.} \and
 Pravesh K. Kothari\thanks{Carnegie Mellon University, \texttt{praveshk@cs.cmu.edu}. Supported by  NSF CAREER Award \#2047933, Alfred P. Sloan Fellowship and a Google Research Scholar Award.}\and
 Lucas Pesenti\thanks{Bocconi University, \texttt{lucas.pesenti@phd.unibocconi.it}.}\and
 Luca Trevisan\thanks{Bocconi University, \texttt{l.trevisan@unibocconi.it}. This project has received funding from the European Research Council (ERC) under the European Union’s Horizon 2020 research and innovation programme (grant agreement No. 834861).}
 }

\date{\today}
\maketitle

\begin{abstract}
We give new rounding schemes for SDP relaxations for the problems of maximizing cubic polynomials over the unit sphere and the $n$-dimensional hypercube. In both cases, the resulting algorithms yield a $O(\sqrt{n/k})$ multiplicative approximation in $2^{O(k)} \poly(n)$ time. In particular, we obtain a $O(\sqrt{n/\log n})$ approximation in polynomial time. For the unit sphere, this improves on the rounding algorithms of~\cite{BGG+17} that need quasi-polynomial time to obtain a similar approximation guarantee. Over the $n$-dimensional hypercube, our results match the guarantee of a search algorithm of Khot and Naor \cite{KN08} that obtains a similar approximation ratio via techniques from convex geometry. Unlike their method, our algorithm obtains an upper bound on the integrality gap of SDP relaxations for the problem and as a result, also yields a \emph{certificate} on the optimum value of the input instance. Our results naturally generalize to homogeneous polynomials of higher degree and imply improved algorithms for approximating satisfiable instances of Max-3SAT.

Our main motivation is the stark lack of rounding techniques for SDP relaxations of higher degree polynomial optimization in sharp contrast to a rich theory of SDP roundings for the quadratic case. Our rounding algorithms introduce two new ideas: 1) a new polynomial reweighting based method to round sum-of-squares relaxations of higher degree polynomial maximization problems, and 2) a general technique to \emph{compress} such relaxations down to substantially smaller SDPs by relying on an explicit construction of certain hitting sets. We hope that our work will inspire improved rounding algorithms for polynomial optimization and related problems. 

\end{abstract}

\section{Introduction}

The focus of this paper is on \emph{polynomial optimization}: approximate the maximum of a given $n$-variate polynomial $p$ over the Boolean hypercube $\pmo^n$ or the unit $n$-dimensional sphere $\calS^{n-1}$. This formulation is very expressive and captures several important discrete and continuous optimization problems. Representative examples include constraint satisfaction problems (CSPs) such as Max-Cut and Max-$3$SAT; the best separable state problem~\cite{BKS17} and the QMA(2) vs EXP conjecture in quantum information~\cite{AIM14}; the problem of finding lowest energy states of spin-glass systems in statistical physics~\cite{SK75}; and the problem of finding optimal Lyapunov certificates of stability in control theory~\cite{Par00}. 

Thanks to spectacular developments in the design and analysis of algorithms using semidefinite programming, the case when $p$ is a homogeneous \emph{quadratic} polynomial is well understood. Over the unit sphere, the problem is easy and is equivalent to computing the maximum eigenvalue of the coefficient matrix. Over the hypercube, Goemans and Williamson~\cite{GW95} ushered in a new era for such problems via their elegant ``hyperplane'' rounding of the canonical semidefinite programming (SDP) relaxation for quadratic optimization. The resulting momentum led to similar rounding algorithms for the bipartite case (a.k.a., the celebrated Grothendieck inequality) \cite{AN06}, 
positive semidefinite case (via the Nesterov rounding)
and arbitrary quadratics (via the Megretski rounding)~\cite{Meg01,CW04}. Some examples of the many algorithmic applications include approximation algorithms for 2-variable CSPs, cut norms of matrices~\cite{AN06}, and correlation clustering~\cite{CW04}. To top it off, a general result of Raghavendra~\cite{Rag08} shows that rounding the canonical SDP is in fact the \emph{optimal} approximation algorithm for all CSPs assuming the Unique Games Conjecture (UGC). Together, these results establish a conjecturally complete picture of quadratic maximization and underscore the centrality of SDP rounding. 

In sharp contrast to this rich picture of the quadratic case, even the very next step of \emph{cubic} optimization is scarcely understood\footnote{We note that for random polynomial optimization (as opposed to the worst-case setting of interest to this work) and related problems, there has been considerable recent success in obtaining new algorithms via rounding sum-of-squares relaxation of polynomial optimization formulations.}.
This can be directly ``blamed'' on the significantly under-developed technology for rounding canonical SDP relaxations of higher degree polynomial optimization problems. The only known general result in this direction is the rounding scheme of Bhattiprolu et.~al.~\cite{BGG+17} for the natural SDP relaxation (arising out of the sum-of-squares hierarchy of semidefinite programs) of (a variant of) polynomial optimization\footnote{The work of Bhattiprolu et.~al.~\cite{BGG+17} considers the variant where one intends to maximize the \emph{absolute value} of a polynomial and this distinction makes a material difference to the difficulty of the problem.} over the unit $n$-dimensional sphere. For cubic polynomials, their algorithm gives an $O(\sqrt{n})$-approximation. Over the hypercube, we know of no rounding algorithms (for any convex relaxation) that obtain a non-trivial approximation guarantee. Indeed, the only known algorithmic result in this direction is the seminal work of Khot and Naor~\cite{KN08} that gives a randomized $O(\sqrt{n/\log n})$-approximation algorithm by circumventing SDP rounding altogether and relying on anti-concentration inequalities and techniques from convex geometry instead. On the flip side, random homogeneous cubic polynomials (over the hypercube, these are equivalent to random $3$-XOR formulas) are known to exhibit a $\wt{\Omega}(n^{1/4})$ integrality gap for the natural SDP relaxation and their constant degree sum-of-squares strengthenings (over both the unit sphere and the hypercube), while the best known NP-hardness~\cite{Has01} can rule out only constant-factor approximation algorithms.  

In addition to being a central algorithmic problem, an improved understanding of higher degree polynomial maximization (in fact just cubic!) is vital to progress on central open questions in disparate areas. This formulation immediately captures the basic problem of maximizing the ``advantage over random assignments''\footnote{For example, a random assignment satisfies a $7/8$ fraction of constraints for any 3SAT instance. Thus, it is natural to ask: given an instance with optimum value $7/8 + \eps$, what's the largest fraction of constraints that we can satisfy in polynomial time? This is a more nuanced formulation of the classical question of approximating Max-3SAT.} for approximating Max-3SAT (and other CSPs). Further, sufficiently strong (and yet far from what known hardness results and integrality gaps rule out) approximation algorithms for special cases of degree-$3$ and degree-$4$ polynomial optimization can refute the Small-Set Expansion~\cite{BBH+12} hypothesis (a close variant of the UGC), settle the Aaronson-Impagliazzo-Moshkovitz~\cite{AIM14} conjecture that relates to the power of quantum entanglement, and refute the celebrated ``planted clique hypothesis''~\cite{FK08}.

To summarize, there are wide gaps in our understanding of the landscape of cubic optimization, and much of it can be attributed to the significant dearth of SDP rounding algorithms for higher degree polynomial maximization. 

\subsection{Our results}

In this paper, we design new rounding algorithms for SDP relaxations of cubic polynomial optimization  and as a consequence, obtain improved guarantees for polynomial optimization over both the unit sphere and the hypercube. Our first result gives a rounding algorithm for the \emph{canonical} sum-of-squares relaxation for homogeneous cubic optimization.

\begin{theorem}[Informal, see \Cref{thm:higher-degree-sos,thm:sos-spherical}] \label{thm:sos-rounding}
    For every $k \geq 6, n \in \N$, there is an $n^{O(k)}$-time rounding algorithm for the canonical degree-$k$ sum-of-squares relaxation of homogeneous cubic maximization problem (over both the unit sphere and the hypercube) that achieves an $O(\sqrt{n/k})$-multiplicative approximation.
\end{theorem}

This relaxation was first analyzed in~\cite{BGG+17} to obtain similar guarantees to \Cref{thm:sos-rounding} but only for the case of the unit sphere. Their rounding algorithm is based on significantly different ideas relying on ``weak decoupling'' inequalities and involve reasoning about eigenvectors of the SDP solution. In particular, such techniques seem to have no natural analogs over the Boolean hypercube. As a result, prior to our work, there were no known bounds on the integrality gap of the sum-of-squares relaxations of the foundational problem of homogeneous polynomial optimization over the hypercube. 

Our techniques in fact extend naturally to homogeneous polynomial optimization for all \emph{odd}-degree polynomials (see \Cref{sec:high-degree}). Handling \emph{even} degree polynomials is related to the inherent (and well-known) issue of non-existence of decoupling inequalities and is an outstanding open question.
Informally speaking, decoupling inequalities (\Cref{lem:decoupling,lem:decoupling-hd}) allow us to assume that the underlying objective function is a \emph{tripartite} polynomial (i.e., the variables can be partitioned into three classes such that each non-zero monomial uses one variable from each class) with only a constant factor difference in the optimum. Such a fact is simply not true for homogeneous even-degree polynomials (see \Cref{ex:failure-of-decouling}). Indeed, this is the key reason why \cite{BGG+17} considers the problem of maximizing the \emph{absolute} value of a polynomial instead, in which case decoupling inequalities are known for all degrees \cite{HLZ10}. 
We note (see \Cref{sec:high-degree}) that for maximizing the absolute value of a polynomial, our techniques naturally extend to homogeneous polynomials of all degrees.

\paragraph{Compressing SDPs.} Our second main result shows that the sum-of-squares relaxation can be \emph{compressed} from $n^{O(k)}$-size to just $2^{O(k)} \poly(n)$-size SDP while preserving the same approximation guarantees as in \Cref{thm:sos-rounding}. This is reminiscent, at a high-level, of the work of Guruswami and Sinop~\cite{GS12} in the context of Max Bisection and of Barak, Raghavendra and Steurer~\cite{BRS11} for Unique Games. Both these examples are closely related to 2-CSPs and thus quadratic polynomials. Our work gives the first analog of such a pruning for higher-degree polynomial optimization.

\begin{theorem}[Informal, see \Cref{thm:pruned-version,thm:pruned-version-sphere}] \label{thm:pruning}
    For every $k, n \in\N$, there is a $2^{O(k)} n^{O(1)}$-time algorithm that takes input a homogeneous multilinear cubic polynomial $f(x)$ in $n$ variables and outputs an assignment that achieves an $O(\sqrt{n/k})$-approximation to the optimum of $f$ over the hypercube or the unit sphere. 
    Moreover, our algorithm is obtained by rounding an SDP relaxation of the cubic optimization problem.
\end{theorem}

\Cref{thm:pruning} improves the approximation guarantees of~\cite{BGG+17}. In particular, we obtain an $O(\sqrt{n/\log n})$ approximation algorithm in polynomial time, as opposed to quasi-polynomial time in~\cite{BGG+17}. In contrast to the prior works on pruning SoS relaxations, the proof of \Cref{thm:pruning} relies on ``derandomizing'' a crucial set of inequalities that naturally arise in our rounding algorithm via polynomial reweightings. 

Over the $n$-dimensional hypercube, an $O(\sqrt{n/\log n})$-approximation algorithm for homogeneous cubic maximization was given by Khot and Naor~\cite{KN08}. Their algorithm relies on ideas from convex geometry and involves a randomized reduction to the classical bipartite quadratic optimization problem. Our algorithm for \Cref{thm:pruning} is \emph{deterministic} and our guarantees on the hypercube in the polynomial-time regime match the ones of \cite{KN08}. 

Finally, we note that while prior works \cite{GS12,BRS11} obtain smaller SDPs by \emph{pruning} constraints and variables from the sum-of-squares SDP relaxations, our compressed SDPs are obtained by \emph{adding} certain auxiliary variables and constraints that are enough to imply the relevant inequalities for the analysis of the rounding algorithm. This is similar to the work of Steurer and Tiegel \cite{ST21} that showed a degree reduction result for some applications in robust statistics where one can take a high degree SoS relaxation and get a new SDP with additional variables but a lower degree with a similar performance. However, the result of \cite{ST21} utilizes problem-specific structure and does not appear to be directly relevant to our applications in this work.

\paragraph{Certification algorithms and application to 3SAT.} As a key consequence of our new rounding algorithms, we obtain \emph{certificates} of upper bounds on the true maximum of a homogeneous cubic objective. This is because rounding algorithms immediately imply upper bounds on the integrality gaps of the underlying SDP relaxations. It is often desirable (both from a practical, e.g.\ SAT solving, and theoretical standpoint, e.g.\ refutation problems in average-case complexity) to ask for algorithms that produce certificates on the optimum of an optimization problem.
We note that the problem of finding certifiable approximation algorithms for cubic maximization was formulated by Trevisan in his blog post~\cite{Trevisan}. We refer the reader to the post for additional discussion.

In \Cref{sec:3-sat}, we demonstrate a concrete advantage of our new certification algorithm by using it to obtain an improved rounding algorithm for \emph{satisfiable} 3SAT instances:

\begin{theorem}
    \label{thm:3sat}
    There is a polynomial-time randomized algorithm that, given a satisfiable 3SAT formula with $n$ variables, finds with high probability an assignment satisfying a $(\frac 7 8 + \wt{\Omega}(n^{-\frac 3 4}))$-fraction of the clauses.
\end{theorem}

For comparison, the best known result of \Hastad and Venkatesh~\cite{HV04} gives $\frac{7}{8} + \Omega\left(\frac 1 {\sqrt{m}}\right)$ (via random assignments), which is in general polynomially worse as $m$ can be $\Omega(n^3)$. On the lower bound side, achieving approximation $\frac{7}{8}+\eps$ for any $\eps > 0$ is NP-hard~\cite{Has01}.

Note that the fraction of satisfied constraints of a 3SAT formula can be written as $\frac{7}{8} + f(x)$, where $f(x)$ is a \emph{non-homogeneous} degree-3 polynomial (the advantage over random assignment).
Thus, due to the degree-2 part of $f$, the naive approach of decoupling and optimizing the degree-3 part of $f$ (e.g., using \Cref{thm:sos-rounding}) does not work.
Our main idea is that even though $f$ has components of degree 1 and 2, they behave nicely when the instance is satisfiable (see \Cref{obs:satisfiable-3sat}).
This allows us to add extra constraints to the sum-of-squares relaxation that enforce the niceness of the degree-2 part (see \Cref{alg:3sat-improved}).
We remark that this is an example of the benefit of designing \emph{certifiable} approximation algorithms using SDPs, as it allows us to add additional constraints for specific applications.

We leave as future work the problem of achieving value $\frac{7}{8} + \Omega\left(\frac 1 {\sqrt{n}}\right)$ for 3SAT, which would match our result for the degree-3 homogeneous case.

\subsection{Technical overview}

\paragraph{Our rounding algorithm.} Our first main idea is a new rounding scheme based on \emph{polynomial reweightings} (introduced in~\cite{BKS17}) that at once replicates the guarantees of \cite{BGG+17} on the unit sphere while naturally extending (in fact, with essentially no change!) to the hypercube.
Let us elaborate on this new scheme of roundings next.

Our rounding works with \emph{decoupled} polynomials --- cubic polynomials of three separate classes of variables $x,y,z$ such that the input polynomial can be written as $f(x,y,z) = \sum_{i=1}^n x_i q_i(y,z)$ for some quadratic polynomial $q_i \coloneqq q_i(y,z) = y^{\top} T_i z$. Via classical decoupling inequalities (\Cref{lem:decoupling}), working with such polynomials is without loss of generality (at the cost of losing only an absolute constant factor in the approximation ratio).

Informally speaking, an SDP solution is described by a linear operator $\pE$ that maps low-degree polynomials into real numbers so that $\pE[p^2] \geq 0$ for every $p$ that is low-degree (see \Cref{sec:sos-prelims}). In other words, the linear operator ``pretends'' like an expectation of a probability distribution over the unit sphere/hypercube for the purpose of taking expectations of low-degree polynomials. 

Consider the value that an SDP solution assigns to the input polynomial $f$. In particular, if it so happens that $\sum_{i} |\pE[q_i]|$ is at least an $O(\sqrt{k/n})$-fraction of the SDP objective, then by using the Grothendieck inequality (\Cref{fact:grothendieck}) to round $y$ and $z$ into integral $\ol{y},\ol{z}$ and setting $x_i$ to be $\sign(q_i(\ol{y},\ol{z}))$, we obtain an $O(\sqrt{k/n})$-approximation to the optimum. Now, this reasoning relies on the SDP solution $\pE$ giving a large value to $\sum_i |q_i|$, that is, $\sum_i |\pE[q_i]|$ being large. 

Of course, there is no reason for such a property to hold for the actual SDP solution we compute. 
Our key idea is to show that there is an efficiently computable transformation from the initial $\pE$ into a new $\pE'$  with the property that 1) $\pE'$ continues to have an objective value as large as the one we began with, and 2) $\sum_i |\pE'[ q_i]|$ is at least an $O(\sqrt{k/n})$-fraction of the SDP objective. 
This transformation is obtained by ``reweighting'' (a generalization of ``conditioning" procedure first used in~\cite{BKS17} for rounding high-degree sum-of-squares SDP relaxations) the original solution $\pE$ by an appropriate sum-of-squares polynomial $r(x)$. In fact, $r(x)$ is chosen to be an appropriate random polynomial in our analysis. Informally speaking, if one thinks of the SDP solution as an actual probability distribution over $x,y,z$, then our reweighting can be thought of as shifting mass around a randomly chosen direction in $x$ to ensure that the $\sum_i |\pE'[q_i]|$ becomes large enough. 

\paragraph{Compression via hitting sets for halfspaces.} Constructing the compressed SDP for \Cref{thm:pruning} involves adding some constraints that enforce a certain appropriate variant of \emph{anti-concentration} inequalities to hold for all projections of the SDP solution. Such inequalities are implied by the constraints in a degree-$k$ sum-of-squares relaxation, but would cost an $n^{O(k)}$-bound in the size of the resulting SDP. Our main observation is that we can obtain the relevant anti-concentration inequalities by explicitly forcing constraints on an appropriate ``hitting set'' for halfspaces. To implement this, we need an efficient construction of a set of $2^k n^{O(1)}$ points from $\{ -1,1\}^n$ (respectively, $2^{O(k)} n^{O(1)}$ elements of the unit sphere) such that for every vector $w$ there is at least one element $x$ of the hitting set such that $\langle x,w \rangle \geq \Omega(\sqrt{k/n})  \|w\|_1$
(respectively, at least one element $x$ such that
$\langle x,w \rangle \geq \Omega(\sqrt {k/n}) \|w\|_2$). The desired condition occurs with probability $2^{-\Theta(k)}$ for a Rademacher (respectively, Gaussian) choice of $x$, so the required hitting set property is satisfied by the support of any pseudorandom generator that $\varepsilon$-fools all linear threshold functions for $\varepsilon = 2^{-\Theta(k)}$. To obtain the compressed SDPs we desire, we need $\varepsilon$-PRGs (or even $\varepsilon$-Hitting Set Generators) for all linear threshold functions with optimal seed length $O(\log n + \log(1/\eps))$. However, the best known seed length for an efficient, explicit construction is $O(\log n + \log^2(1/\eps))$~\cite{MZ10}, which amounts to a hitting set of size $n^{O(\log n)}$ in our setting. Instead, we give an elementary construction of a set as above with nearly optimal size that yields us our compressed SDP.

\subsection{Organization of the paper}

In \Cref{sec:prelims}, we introduce various tools and facts needed in the sequel. We finish that section by reformulating the algorithm of Khot and Naor \cite{KN08} to pave the way for the presentation of our algorithmic techniques.

In \Cref{sec:rounding-main}, we present our rounding algorithms for the canonical degree-$k$ SoS relaxation of cubic optimization over the hypercube.

In \Cref{sec:compressed}, we state and analyze our compressed SDP relaxation that achieves better approximation ratio for cubic optimization over the hypercube.

In \Cref{sec:sphere}, we adapt the techniques to show the analogous results over the unit sphere.

In \Cref{sec:3-sat}, we illustrate the power of certification algorithms by giving an improved approximation for Max-3SAT in the satisfiable regime.

In \Cref{sec:warmup}, we describe an alternative simple certification algorithm that achieves approximation $O(\sqrt n)$ for cubic optimization over the hypercube.

In \Cref{sec:high-degree}, we show how to generalize our techniques to higher-degree polynomials.

\section{Preliminaries}
\label{sec:prelims}

\subsection{Sum-of-Squares SDPs} 
\label{sec:sos-prelims}

We refer the reader to the monograph~\cite{FKP19} and the lecture notes~\cite{BS16} for a detailed exposition of the sum-of-squares method and its usage in algorithm design. 

A \emph{degree-$\ell$ pseudo-distribution} $\mu$ over variables $x_1, x_2, \ldots, x_n$ corresponds to a linear operator $\pE_{\mu}$ that maps polynomials of degree $\leq \ell$ to real numbers and satisfies $\pE_\mu\, 1=1$ and $\pE_{\mu}[p^2] \geq 0$ for every polynomial $p(x_1,x_2, \ldots, x_n)$ of degree $\leq \ell/2$. We say that such a pseudo-distribution satisfies the \emph{hypercube} constraints if $\pE_{\mu}[p x_i^2] = \pE_{\mu}[p]$ for every polynomial $p$ of degree $\leq \ell-2$ and every $i\in [n]$. We say that such a pseudo-distribution satisfies the \emph{unit sphere} constraints if $\pE_{\mu}[ \norm{x}_2^2 p] =  \pE_{\mu}[p]$ for every $p$ of degree $\leq \ell-2$. 

Given a polynomial $p$ (with the $\ell_1$-norm of the coefficients being $\norm{p}_1$) over $x_1, x_2,\ldots,x_n$, a pseudo-distribution of degree $\ell$ over the unit sphere or the hypercube that maximizes $p$ within an additive $\varepsilon \norm{p}_1$ error can be found in time $n^{O(\ell)} \polylog(n/\varepsilon)$ via the ellipsoid method. 

\paragraph{Reweighting.} Given a pseudo-distribution $\mu$ over the unit sphere or the hypercube, a \emph{reweighting} of $\mu$ by a sum-of-squares polynomial $q$ satisfying $\pE_{\mu}[q]>0$ is a pseudo-distribution $\mu'$ that maps any polynomial $p$ to $\pE_{\mu'}[p] = \pE_{\mu}[pq]/\pE_{\mu}[q]$. For any $\mu$ of degree $\ell$ and $q$ of degree $ r < \ell$, $\mu'$ is a pseudo-distribution of degree at least $\ell-r$. Furthermore, if $\mu$ satisfies the unit sphere (or the hypercube) constraints then so does $\mu'$ as long as $r\leq \ell- 2$.

\paragraph{Sum-of-squares proofs.}
Let $f_1,f_2,\dots,f_m$ and $g$ be multivariate polynomials in $x$.
A \emph{sum-of-squares} proof that the constraints $\{f_1 \geq 0,\dots, f_m \geq 0\}$ imply $g \geq 0$ consists of sum-of-squares polynomials $(p_S)_{S\subseteq[m]}$ such that $g = \sum_{S\subseteq[m]} p_S \prod_{i\in S} f_i$.
The \emph{degree} of such a sum-of-squares proof equals the maximum of the degree of $p_S \prod_{i \in S} f_i$ over all $S$ appearing in the sum above. We write $\{f_i \geq 0,\ \forall i\in[m] \} \sststile{t}{x} \{g \geq 0\}$ where $t$ is the degree of the sum-of-squares proof.

We will rely on the following basic connection between SoS proofs and pseudo-distributions:

\begin{fact}
Suppose $\{f_i \geq 0, \forall i\in[m]\} \sststile{t}{} \{g \geq 0\}$ for some polynomials $f_i$ and $g$.
Let $\mu$ be a pseudo-distribution  of degree $\geq t$ satisfying $\{f_i \geq 0\}_{i\in[m]}$. Then, $\pE_\mu[g] \geq 0$.
\end{fact}

We next state some standard facts (see \cite{FKP19} for references).

\begin{fact}[SoS generalized triangle inequality] \label{fact:sos-triangle}
    Let $k \in \N$ and $x = (x_1,\dots,x_n)$ be indeterminates.
    \begin{equation*}
        \Set{x_i \geq 0,\ \forall i \in [n]} \sststile{k}{x_1,\dots,x_n} \Set{ \Paren{\frac{1}{n}\sum_{i=1}^n x_i}^k \leq \frac{1}{n} \sum_{i=1}^n x_i^k }\mper
    \end{equation*}
    Moreover, if $k$ is even, then
    \begin{equation*}
        \sststile{k}{x_1,\dots,x_n} \Set{ \Paren{\frac{1}{n}\sum_{i=1}^n x_i}^k \leq \frac{1}{n} \sum_{i=1}^n x_i^k } \mper
    \end{equation*}
\end{fact}
Note that the $k=2$ case is the SoS version of the Cauchy-Schwarz inequality for vectors.

\begin{fact}[Cauchy-Schwarz for pseudo-distributions]
\label{fact:sos-cauchy-schwarz}
    Let $f, g$ be polynomials of degree at most $d$ in indeterminate $x$.
    Then, for any degree-$2d$ pseudo-distribution $\mu$ over $x$, we have
    \[
        \pE_{\mu}[fg] \leq \sqrt{\pE_{\mu}[f^2]} \sqrt{\pE_{\mu}[g^2]}\mper
    \]
\end{fact}

\begin{fact}[\Holder's inequality for pseudo-distributions]
\label{fact:sos-holders}
    Let $f,g$ be polynomials of degree at most $d$ in indeterminate $x$, and fix any even $t\in \N$.
    Then, for any degree-$td$ pseudo-distribution $\mu$ over $x$, we have
    \[
        \pE_{\mu}[f^{t-1} g] \leq \paren{\pE_{\mu}[f^t]}^{\frac{t-1}{t}} \paren{\pE_{\mu}[g^t]}^{\frac{1}{t}}\mper
    \]
    Furthermore, for any $t\in \N$ and any degree-$2td$ pseudo-distribution $\mu$, we have $\pE_{\mu}\left[f^{2t-2}\right] \leq \pE_{\mu}\left[f^{2t}\right]^{\frac{t-1}{t}}$. 
\end{fact}

\subsection{Basic rounding algorithms}

We state here three standard rounding algorithms for the quadratic case.

\begin{fact}[Grothendieck rounding~\cite{AN06}]
\label{fact:grothendieck}
    There is a rounding algorithm that, given any degree-2 pseudo-distribution $\mu$ over $x,y\in \pmo^n$, outputs $\ol{x}, \ol{y}\in \pmo^n$ such that for any matrix $M\in \R^{n\times n}$,
    $\ol{x}^\top M \ol{y} \geq \frac{1}{K_G} \cdot \pE_{\mu}\left[x^\top M y\right]$,
    where $K_G < 1.783$ is the Grothendieck constant.
\end{fact}

\begin{lemma}[Lossless rounding on the unit sphere] \label{lem:rounding-sphere}
    Given any degree-2 pseudo-distribution $\mu$ over $x \in \calS^{n-1}$, there is a rounding algorithm that outputs $\ol{x} \in \calS^{n-1}$ such that for any matrix $M\in \R^{n\times n}$,
    $\ol{x}^\top M \ol{x} \geq \pE_{\mu}\left[x^\top M x\right]$.
\end{lemma}
\begin{proof}
    Let $X = \pE_{\mu}\left[xx^\top\right]$ and let $X = \sum_{i=1}^n \lambda_i v_i v_i^\top$ be its eigen-decomposition with $\lambda_i \geq 0$ and $v_i \in \calS^{n-1}$ for all $i\in[n]$.
    Since $\mu$ is over the unit sphere, $\tr(X) = \sum_{i=1}^n \lambda_i = 1$, hence $\{\lambda_i\}_{i\in[n]}$ defines a valid probability distribution.
    Now if we sample $\ol{x} = v_i$ with probability $\lambda_i$ for all $i\in [n]$, we get
    $\E[\ol{x}^\top M \ol{x}] = \sum_{i=1}^n \lambda_i v_i^\top M v_i = \iprod{X, M} = \pE_{\mu}\left[x^\top M x\right]$.
    In particular one of $v_1, \ldots, v_n$ must have quadratic form at least $\pE_\mu\left[x^\top Mx\right]$.
\end{proof}

\begin{theorem}[Charikar-Wirth rounding \cite{CW04}]
    \label{thm:charikar-wirth}
    Given any degree-2 pseudo-distribution $\mu$ over $x\in\{\pm 1\}^n$, for any $T > 0$, there is a polynomial-time sampleable distribution $\mathcal D$ supported on $\{\pm 1\}^n$ such that for any matrix $M\in \R^{n\times n}$ with zero diagonal entries,
    \[
        \E_{\ol{x}\sim \calD} \left[\ol{x}^\top M \ol{x}\right]
        \ge \frac{1}{T^2} \cdot \pE_\mu \left[x^\top M x\right] - 8e^{-\frac{T^2}{2}} \sum_{i,j\in[n]} |M_{ij}| \mper
    \]
    In particular, if $\mu$ is an optimal pseudo-distribution for the degree-2 SoS relaxation of $\max_{x\in\pmo^n} x^\top Mx$, then by picking $T=\Theta(\sqrt{\log n})$,
    \[
        \E_{\ol{x}\sim \calD} \left[\ol{x}^\top M \ol{x}\right]\ge \Omega\left(\frac{1}{\log n}\right) \max_{x\in \pmo^n} x^\top Mx\mper
    \]
\end{theorem}

\subsection{Anti-concentration}

We will need the following anti-concentration result that can be deduced from standard hypercontractivity (see e.g. \cite[Lemma 3.2]{AGK04}).

\begin{lemma}
    \label{lem:anti-concentration-hypercontractivity}
    Let $\calD$ be a distribution over $\R^n$ satisfying the following: there exists a constant $B > 0$ such that $\E_{x\sim \calD} [p(x)^4] \leq B^d \cdot \E_{x\sim \calD}\left[p(x)^2\right]^2$ for every degree-$d$ polynomial $p$.
    Then, for any degree-$d$ polynomial $p$,
    \begin{equation*}
        \Pr_{x\sim \calD} \Brac{p(x) > \E_\calD\, p} \geq 2^{-\frac{4}{3}} B^{-d} \mper
    \end{equation*}
\end{lemma}

Relevant special cases for what follows are when $\calD$ is the uniform distribution over $\pmo^n$ or the standard Gaussian distribution $\calN(0,I_n)$, which both satisfy the assumption with $B = 9$ (see e.g.\ for reference \cite[Theorem 9.21]{O14} and \cite[Theorem 1.6.2]{Bog98}).

\subsection{Decoupling}

A standard technique for polynomial optimization problems is  \emph{decoupling}, which relates the optimum of $\iprod{T, x^{\ot 3}}$ (the ``coupled'' polynomial) to the optimum of $\iprod{T, x \otimes y \otimes z}$ (the ``decoupled'' polynomial). For polynomial optimization over the $n$-dimensional hypercube or the unit sphere, these two quantities are within a constant factor from each other when the polynomial has \textit{degree~3}.
Thus, in the rest of the paper, we will assume that the given polynomial is already decoupled.

\begin{lemma}[Decoupling \cite{KN08,HLZ10}]
    \label{lem:decoupling}
    Let $\Omega$ be either $\{-1,1\}^n$ or $\calS^{n-1}$.
    Consider a multilinear homogeneous degree-3 polynomial in $n$ variables $f(x) = \sum_{i,j,k=1}^n T_{ijk} x_i x_j x_k$ (where $T$ is a symmetric $3$-tensor). Consider also  the decoupled version of $f$: $\wt{f}(x,y,z) = \sum_{i,j,k=1}^n T_{ijk} x_i y_j z_k$.
    Then,
    \[
        \max_{x\in \Omega} f(x)
        \geq \frac{2}{9} \cdot \max_{x,y,z\in \Omega} \wt f(x,y,z) \mper
    \]
\end{lemma}

In \Cref{sec:high-degree}, we prove a generalization of \Cref{lem:decoupling} to all \textit{odd-degree} polynomials.

\subsection{The algorithm of Khot and Naor}
\label{sec:khotnaor}

In this section, we describe the algorithm of Khot and Naor~\cite{KN08} yielding an approximation of $O(\sqrt{n/\log n})$ to the maximum of $f(x,y,z)=\sum_{i,j,k=1}^n T_{ijk} x_i y_j z_k$ over $x,y,z\in\pmo^n$.
The original algorithm was done via a reduction to estimating the $L_1$-diameter of some convex body.
We reformulate the algorithm in a more direct way, which also serves as inspiration for our algorithms in \Cref{sec:rounding-main}:

\begin{enumerate}
    \item Sample $\ol{x} \sim \pmo^n$. Let $f_{\ol{x}}(y,z) \coloneqq \iprod{T, \ol{x} \ot y \ot z} = \sum_{i=1}^n \ol{x}_i (y^\top T_i z)$ be the (decoupled) degree-2 polynomial in variables $y$ and $z$. Here $T_i$ is a slice of the tensor $T$, that is, an $n\times n$ matrix.
    \item Using Grothendieck rounding (\Cref{fact:grothendieck}), find $\ol{y},\ol{z} \in \pmo^n$ which achieve a constant factor approximation for $\max_{y,z\in \pmo^n} f_{\ol{x}}(y,z)$.
    \item Finally, repeat steps 1 and 2 $\poly(n)$ times and output the best solution $(\ol{x}, \ol{y}, \ol{z})$ obtained.
\end{enumerate}

The key lemma to analyze this algorithm is the following anti-concentration inequality.
\begin{lemma}[Lemma 3.2 of \cite{KN08}] \label{lem:KN-anticoncentration}
    For any $\delta \in (0,1/2)$, there is a constant $c(\delta) > 0$ such that for any $a \in \R^n$, if $\eps_1,\dots,\eps_n$ are i.i.d.\ $\{\pm 1\}$ random variables, then
    \begin{align*}
        \Pr\Brac{\sum_{i=1}^n a_i \eps_i \geq \sqrt{\frac{\delta \log n}{n}} \cdot \|a\|_1} \geq \frac{c(\delta)}{n^{\delta}} \mper
    \end{align*}
\end{lemma}

The proof of Khot and Naor's result  easily follows  from \Cref{lem:KN-anticoncentration}.
Indeed, let $x^*, y^*, z^* \in \pmo^n$ be an optimal solution.
First observe that it is always optimal to set $x^*_i \coloneqq \sgn(\iprod{T_i, y^* \ot z^*})$, so that $\OPT\coloneqq f(x^*,y^*,z^*) = \sum_{i=1}^n |\iprod{T_i, y^* \ot z^*}|$. Then, for any $\ol{x}\in\pmo^n$, the algorithm outputs $(\ol{x}, \ol{y}, \ol{z})$ such that
\begin{align*}
    f(\ol{x}, \ol{y}, \ol{z}) \geq \Omega(1) \cdot \max_{y,z\in \pmo^n} f(\ol{x}, y,z)
    \geq \Omega(1) \cdot f(\ol{x}, y^*, z^*) \mper
\end{align*}
However, by \Cref{lem:KN-anticoncentration}, with at least inverse polynomial probability, a random $\ol{x}$ satisfies 
\[
    f(\ol{x}, y^*, z^*) = \sum_{i=1}^n \ol{x}_i \iprod{T_i, y^* \ot z^*} \geq \Omega(1)\cdot\sqrt{\frac{\log n}{n}} \cdot \sum_{i=1}^n |\iprod{T_i, y^* \ot z^*}| = \Omega(1)\cdot\sqrt{\frac{\log n}{n}} \cdot \OPT\mper
\]
Thus, by repeating $\poly(n)$ times, with high probability the algorithm outputs an assignment that has value $\Omega\Bigparen{\sqrt{\frac{\log n}{n}}} \cdot \OPT$.

\section{Rounding SoS relaxations for cubic optimization}
\label{sec:rounding-main}

In this section, we present our rounding algorithms for the canonical degree-$k$ SoS relaxation of cubic optimization over the hypercube. 

\begin{itemize}
    \item In \Cref{sec:sqrtn}, we show how to achieve approximation $O(\sqrt{n})$ by rounding the canonical degree-$6$ SoS relaxation via polynomial reweightings.
    
    \item In \Cref{sec:approx-higher}, we extend the previous argument to achieve approximation $O(\sqrt{\frac{n}{k}})$ by rounding the canonical degree-($6k$) SoS relaxation. The additional key idea is a sum-of-squares proof of a certain anti-concentration argument.
\end{itemize}

\subsection{An \texorpdfstring{$O(\sqrt{n})$}{O(sqrt(n))}-factor approximation}
\label{sec:sqrtn}

We start by giving a simple polynomial-time certification and rounding algorithm via constant-degree SoS that achieves approximation $O(\sqrt{n})$ for cubic optimization over the hypercube. Our key technical ingredient is a new use of polynomial reweightings of pseudo-distributions (see \Cref{lem:scalar-fixing}).

\begin{theorem} \label{thm:deg6-sos-approx}
    For any decoupled homogeneous degree-3 polynomial $f(x,y,z) = \sum_{i,j,k=1}^n T_{ijk} x_i y_j z_k$, the degree-6 SoS relaxation of $\max_{x,y,z\in\{\pm 1\}^n} f(x,y,z)$
    has integrality gap at most $O(\sqrt{n})$.
    
    Furthermore, there is a polynomial-time rounding algorithm that, given a degree-6 pseudo-distribution $\mu$ with $\SOS \coloneqq \pE_{\mu} f(x,y,z)>0$, outputs a solution $\ol x,\ol y,\ol z\in \pmo^n$ with value $f(\ol x,\ol y,\ol z)\ge \Omega\Paren{\frac{\SOS}{\sqrt{n}}}$.
\end{theorem}

We describe another algorithm outputting a certificate with similar guarantees in \Cref{sec:warmup}. In comparison, the proof in this section will come together with a rounding and will allow us to build up towards a more general tradeoff between time and approximation in \Cref{sec:approx-higher}.

Recall from \Cref{sec:khotnaor} that the strategy of \cite{KN08} is to first sample $\ol x\sim \pmo^n$ and then solve for $\ol y$ and $\ol z$ using Grothendieck rounding.
One might expect that a similar strategy works to round an optimal SoS solution.
However, given an optimal pseudo-distribution $\mu$, it is not clear how $\pE_{\mu} \sum_{i=1}^n \ol x_i (y^\top T_i z)$ relates to the SoS value $\pE_{\mu} \sum_{i=1}^n x_i (y^\top T_i z)$.
In fact, the former can be much smaller than the latter or even zero.

Denote $q_i(y,z) \seteq y^\top T_i z$ for convenience.
Our key idea is that even though $\pE_{\mu} \iprod{\ol{x}, q}$ may be small, we can reweight the pseudo-distribution $\mu$ and get another pseudo-distribution $\mu'$ such that $\pE_{\mu'} \iprod{\ol{x}, q} \gtrsim (\pE_{\mu} \iprod{\ol{x}, q}^2)^{1/2}$. Furthermore, the quantity on the right-hand side can be related to the SoS value (for a typical $\ol{x}$).
One may view this procedure as raising the (pseudo-) expectation of a random variable to be close to its (pseudo-) standard deviation, which is reminiscent of the scalar fixing lemma of \cite{BKS17}.

We capture this idea in the following lemma:

\begin{lemma} \label{lem:scalar-fixing}
    Let $p(x_1,\ldots,x_n)$ be a degree-$t$ polynomial and let $\mu$ be a degree-$3t$ pseudo-distribution over $(x_1,\ldots,x_n)$.
    There is reweighting of $\mu$ by a degree-$2t$ polynomial such that the resulting degree-$t$ pseudo-distribution $\mu'$ satisfies
    \begin{equation*}
        \Abs{\pE_{\mu'}[p]} \geq \frac{1}{3} \cdot \sqrt{\pE_{\mu} \left[p^2\right]} \mper
    \end{equation*}
\end{lemma}

\begin{proof}
    Let $m \seteq \sqrt{\pE_{\mu}\left[p^2\right]} > 0$. We can assume that $\left|\pE_\mu [p]\right| < \frac{m}{3}$, otherwise we are done without any reweighting.
    
    First, suppose that $\left|\pE_\mu \left[p^3\right]\right|\ge \frac{m^3}{3}$. Reweight $\mu$ by the degree-$2t$ SoS polynomial $p^2$ and let $\mu'$ be the resulting pseudo-distribution. Then we have $\left|\pE_{\mu'}[p]\right| = \left|\frac{\pE_\mu\left[p^3\right]}{\pE_\mu\left[p^2\right]}\right|\ge \frac{1}{3}\cdot \sqrt{\pE_\mu\left[p^2\right]}$.
    
    Now suppose that $\left|\pE_\mu \left[p^3\right]\right|< \frac{m^3}{3}$. Reweight $\mu$ by the degree-$2t$ SoS polynomial $(p+m)^2$ and let $\mu'$ be the resulting pseudo-distribution. Note that:
    \[
          \pE_{\mu}\left[(p+m)^2\right]=2m^2+2m\cdot \pE_{\mu}\left[p\right]\in \left[\frac{4m^2}{3},\frac{8m^2}{3}\right]\mcom
    \]
    where we used $\left|\pE_\mu [p] \right| < \frac{m}{3}$. In particular, we are reweighting by a polynomial with non-zero pseudo-expectation, so this is a well-defined operation. Moreover,
    \[
        \pE_{\mu} \left[(p+m)^2 p\right] \ge  2m\cdot\pE_{\mu}\left[p^2\right] - \Abs{\pE_{\mu}\Brac{p^3}} - m^2\cdot \Abs{ \pE_\mu[p] } \geq \frac{4m^3}{3} \mper
    \]
    Putting everything together, we obtain $\Abs{\pE_{\mu'} [p]}\ge \frac{1}{2}\cdot \sqrt{\pE_\mu\left[p^2\right]}$. 
    
    Thus, we get the desired result in both cases.
\end{proof}

We are now ready to prove \Cref{thm:deg6-sos-approx}.

\begin{proof}[Proof of \Cref{thm:deg6-sos-approx}]
    Let $q_i(y,z) \coloneqq y^\top T_i z$ for each $i\in [n]$.
    For simplicity of notation, we will drop the dependence on $y,z$ and denote $q = (q_1,\dots,q_n)$.
    Then, we have
    \begin{equation*}
        \SOS = \sum_{i=1}^n \pE_{\mu}[x_i q_i]
        \leq \sum_{i=1}^n \sqrt{\pE_{\mu}[q_i^2]}
        \leq \sqrt{n \cdot \sum_{i=1}^n \pE_{\mu}[q_i^2]}
        = \sqrt{n\cdot \pE_{\mu} \|q\|_2^2}
        \mcom
    \end{equation*}
    by Cauchy-Schwarz and its pseudo-expectation version (\Cref{fact:sos-cauchy-schwarz}).
    Next, since $\E_{h\sim \pmo^n} \iprod{v, h}^2 = \|v\|_2^2$ is a polynomial identity,
    \begin{equation*} \label{eq:sos-squared}
        \SOS^2 \leq n \cdot \E_{h\sim\pmo^n} \, \pE_{\mu}\, \iprod{q, h}^2 \mcom
        \numberthis
    \end{equation*}
    where we recall that $q = (q_1,\dots,q_n)$ are degree-2 polynomials in $y,z$.
    We now describe the rounding algorithm.
    \begin{enumerate}
        \item Sample $h\sim \pmo^n$, and set $\ol x \coloneqq h$.

        \item Reweight the pseudo-distribution $\mu$ via \Cref{lem:scalar-fixing} to get a degree-2 pseudo-distribution $\mu'$ such that $\left|\pE_{\mu'} \iprod{q, h}\right| \geq \frac{1}{3} \sqrt{\pE_{\mu} \iprod{q,h}^2}$.

        \item Use Grothendieck rounding (\Cref{fact:grothendieck}) on $\mu'$ to obtain solutions $\ol y,\ol z\in \pmo^n$ satisfying $\ol y^\top (\sum_{i=1}^n h_i T_i) \ol z\ge \frac{1}{K_G}\cdot \left|\pE_{\mu'} \iprod{q, h}\right|$ (we can get the guarantees with the absolute value by flipping the sign of $h$).
    \end{enumerate}
    First, note that $\pE_{\mu} \iprod{q,h}^2$ is a degree-2 polynomial in $h$, so by Paley-Zygmund inequality,
    \begin{equation*}
        \Pr_{h\sim \pmo^n} \Brac{ \pE_{\mu}\iprod{q,h}^2 \geq \frac{\SOS^2}{n} } \geq \Omega(1) \mcom
    \end{equation*}
    meaning that  we get a ``good'' $h$ with constant probability. For a good $h$, it also holds that
    \[
        f(\ol x,\ol y,\ol z)\ge \frac{1}{K_G} \left|\pE_{\mu'} \iprod{q, h}\right|\ge \frac{1}{3K_G} \sqrt{\pE_{\mu} \iprod{q,h}^2}\ge \Omega\left(\frac{\SOS}{ \sqrt{n}}\right)\mper
    \]
    Thus, repeating the above $\poly(n)$ times, we can obtain a solution with value $\Omega\Paren{\frac{\SOS}{\sqrt{n}}}$ with high probability.
    This completes the proof.
\end{proof}

\subsection{Going beyond \texorpdfstring{$O(\sqrt{n})$}{O(sqrt(n))}-approximation via higher-degree SoS}
\label{sec:approx-higher}

We now switch to a general time/approximation tradeoff for the problem by rounding higher levels of the SoS hierarchy.

\begin{theorem} \label{thm:higher-degree-sos}
    Let $k,n$ be integers such that $2\le k\le n$.
    For any decoupled homogeneous degree-3 polynomial $f(x,y,z) = \sum_{i,j,k=1}^n T_{ijk} x_i y_j z_k$, the canonical degree-$(6k)$ SoS relaxation of $\max_{x,y,z\in\pmo^n} f(x,y,z)$ has integrality gap at most $O \Bigparen{\sqrt{\frac{n}{k}}}$.
    
    Furthermore, there is an $n^{O(k)}$-time rounding algorithm that, given a degree-$(6k)$ pseudo-distribution with $\SOS \coloneqq \pE_{\mu} f>0$, outputs a solution $\ol{x},\ol{y},\ol{z}\in \pmo^n$ with value $f(\ol{x}, \ol{y}, \ol{z}) \geq \Omega \Bigparen{\sqrt{\frac{k}{n}}} \cdot \SOS$.
\end{theorem}

Recall that the $\sqrt{n}$ approximation factor in the previous section was coming from relating the SoS value $\pE_\mu \iprod{x,q}$ to a quantity of the form $\E_{h\sim \{\pm 1\}^n} \pE_\mu \iprod{h,q}^2$. To make use of higher levels of the SoS hierarchy, we will now connect the SoS value to higher moments of the form $\E_{h\sim \{\pm 1\}^n} \pE_\mu \iprod{h,q}^{2k}$. The proof of \Cref{thm:higher-degree-sos} will then follow from a high-degree version of the  polynomial reweighting from \Cref{lem:scalar-fixing}.

One can interpret the inequality from the previous section $\pE_\mu \iprod{x,q}^2\le n \cdot \E_{h\sim \{\pm 1\}^n} \pE_\mu \iprod{h,q}^2$ as the SoS analog of the inequality $\|q\|_1^2\le n\cdot \|q\|_2^2 = n\cdot \E_{h\sim \{\pm 1\}^n} \iprod{h,q}^2$ that holds for any $q\in\R^n$ by Cauchy-Schwarz and an explicit variance equality. Our higher-level proof also has a classical analog, namely:
\begin{equation}
    \|q\|_1\le O(1)\cdot \sqrt{\frac{n}{k}} \cdot \left(\E_{h\sim \{\pm 1\}^n} \iprod{h,q}^{2k}\right)^{\frac{1}{2k}}\mper\label{eq:high-degree-poly-ineq}
\end{equation}
Such an inequality holds for any $q\in\R^n$ \cite{Mon90}. To see that,  decompose $q$ and $h$ into $k$ (arbitrary) blocks $q^{(1)},\ldots,q^{(k)}$ and $h^{(1)},\ldots,h^{(t)}$ of size roughly $\frac{n}{k}$. By Paley-Zygmund inequality, we get that $\Abs{\iprod{q^{(i)},h^{(i)}}}\ge \Omega(1)\cdot \left\|q^{(i)}\right\|_2\ge \Omega(1)\cdot \sqrt{\frac{k}{n}}\left\|q^{(i)}\right\|_1$ holds with at least constant probability for any fixed $i\in [k]$. So with probability at least $2^{-O(k)}$ we have $\Abs{\iprod{q,h}}\ge \Omega(1)\cdot \sqrt{\frac{k}{n}} \|q\|_1$, which in turn implies \Cref{eq:high-degree-poly-ineq}.

Although this proof is streamlined, the part using Paley-Zygmund and independence across the $k$ blocks does not directly translate into a sum-of-squares proof. We now give a different and degree-$O(k)$ sum-of-squares proof of the inequality.

\begin{lemma} \label{lem:h-v-moment}
    Let $k < n \in \N$, and let $x= (x_1,\dots,x_n)$ be indeterminates and $v = (v_1,\dots,v_n)$ be such that each $v_i$ is a polynomial of degree $\leq t$. Then,
    \begin{equation*}
        \Set{x_i^2=1,\ \forall i \in [n]} \ \sststile{2(t+1)k}{x,v} \
        \E_{h\sim \pmo^n} \Brac{\iprod{h, v}^{2k}} \geq \Paren{\frac{k}{4n}}^k \iprod{x,v}^{2k} \mper
    \end{equation*}
\end{lemma}
\begin{proof}
    We divide $[n]$ into $k$ blocks, each of size at most $\ceil{\frac{n}{k}}$.
    For $t\in [k]$, let $x^{(t)}, v^{(t)}$ be the vectors $x,v$ restricted to the $t$-th block. Then,
    \begin{equation*}
        \sststile{2(t+1)k}{x,v}\ \iprod{x,v}^{2k} = \Paren{\sum_{t=1}^k \iprod{x^{(t)}, v^{(t)}} }^{2k}
        \leq 2^k \cdot \E_{\eps \sim \pmo^k} \Paren{ \sum_{t=1}^k \eps_t \iprod{x^{(t)}, v^{(t)}} }^{2k} \mcom
        \numberthis \label{eq:x-v-2k}
    \end{equation*}
    since $\paren{ \sum_{t=1}^k \eps_t \iprod{x^{(t)}, v^{(t)}} }^{2k}$ is a square for each $\eps \in \pmo^k$.
    Expanding the above and using the fact that all odd moments of $\eps \sim \pmo^k$ vanish, we get
    \begin{equation*}
        \E_{\eps \sim \pmo^k} \Paren{ \sum_{t=1}^k \eps_t \iprod{x^{(t)}, v^{(t)}} }^{2k} = \sum_{\gamma\in \N^k: |\gamma| = k} c_{\gamma} \prod_{t=1}^k \iprod{x^{(t)}, v^{(t)}}^{2\gamma_t} \mcom
        \numberthis \label{eq:expand-E-h}
    \end{equation*}
    where $c_{\gamma} \seteq \frac{(2k)!}{\prod_{t=1}^k (2\gamma_t)!}$.
    Here $|\gamma| = \sum_{t=1}^k \gamma_t$ and $\gamma$ represents a multiset of $[k]$ of size $|\gamma|$.
    Next, by SoS Cauchy-Schwarz (\Cref{fact:sos-triangle}), we have that 
    \begin{equation*}
        \Set{x_i^2=1,\ \forall i \in [n]} \ \sststile{2(t+1)}{x,v}\ 
        \iprod{x^{(t)}, v^{(t)}}^2
        \leq \|x^{(t)}\|_2^2 \cdot \|v^{(t)}\|_2^2
        \leq \Paren{\frac{2n}{k}} \cdot \|v^{(t)}\|_2^2\mcom
    \end{equation*}
    since $x^{(t)}$ has dimension at most $\ceil{\frac{n}{k}} \leq \frac{2n}{k}$.
    Next, using the identity $\|v^{(t)}\|_2^2 = \E_{h^{(t)}}\iprod{h^{(t)}, v^{(t)}}^2$ where $h^{(t)} \sim \pmo^{\dim(x^{(t)})}$,
    \begin{align*}
        \Set{x_i^2=1,\ \forall i \in [n]} \ \sststile{2(t+1)k}{x,v}\ 
        \prod_{t=1}^k \iprod{x^{(t)}, v^{(t)}}^{2\gamma_t}
        &\leq \prod_{t=1}^k \Paren{\frac{2n}{k}}^{\gamma_t} \|v^{(t)}\|_2^{2\gamma_t}\\
        &= \Paren{\frac{2n}{k}}^k \prod_{t=1}^k \Paren{\E_{h^{(t)}} \iprod{h^{(t)}, v^{(t)}}^2}^{\gamma_t} \\
        &\leq \Paren{\frac{2n}{k}}^k \prod_{t=1}^k \E_{h^{(t)}} \iprod{h^{(t)}, v^{(t)}}^{2\gamma_t} \mcom
    \end{align*}
    where the last inequality uses \Cref{fact:sos-triangle}.
    Combining the above with \Cref{eq:x-v-2k,eq:expand-E-h}, we have
    \begin{equation*}
        \Set{x_i^2=1,\ \forall i \in [n]} \ \sststile{2(t+1)k}{x,v}\
        \iprod{x,v}^{2k} \leq \Paren{\frac{4n}{k}}^k \cdot \E_{h\sim \pmo^n} \sum_{\gamma\in \N^k: |\gamma|=k} c_{\gamma} \prod_{t=1}^k \iprod{h^{(t)}, v^{(t)}}^{2\gamma_t} \mper
    \end{equation*}
    Finally, since $h^{(1)},\dots,h^{(k)}$ are uniformly random Boolean vectors, multiplying $h^{(t)}$ by $\eps_t \sim \pmo$ does not change the distribution.
    Thus, applying \Cref{eq:expand-E-h} again, we get
    \begin{align*}
        \E_{h} \sum_{\gamma\in \N^k: |\gamma|=k} c_{\gamma} \prod_{t=1}^k \iprod{h^{(t)}, v^{(t)}}^{2\gamma_t}
        = \E_h \E_{\eps} \Paren{\sum_{t=1}^k \eps_t \iprod{h^{(t)}, v^{(t)}}}^{2k}
        = \E_h \Paren{\sum_{t=1}^k \iprod{h^{(t)}, v^{(t)}}}^{2k}
        = \E_h \iprod{h, v}^{2k} \mper
    \end{align*}
    This completes the proof.
\end{proof}

Our second key ingredient is the analog of \Cref{lem:scalar-fixing} for high moments.

\begin{lemma} \label{lem:scalar-fixing-high-degree}
    Let $k\in \N$.
    Let $p$ be a degree-$t$ polynomial in variables $x\in \R^n$, and let $\mu$ be a degree-$(2k+2)t$ pseudo-distribution.
    There is a degree-$2kt$ reweighting of $\mu$ such that the resulting pseudo-distribution $\mu'$ satisfies
    \begin{equation*}
        \Abs{\pE_{\mu'}p} \geq \frac{1}{3} \cdot \left(\pE_{\mu} \left[p^{2k}\right]\right)^{\frac{1}{2k}} \mper
    \end{equation*}
\end{lemma}

\begin{proof}
    Let $m \seteq \Paren{\pE_{\mu}\left[p^{2k}\right]}^{\frac{1}{2k}} > 0$. 

    First, consider reweighting $\mu$ by the degree-$2kt$ sum-of-squares polynomial $p^{2k}$ and denote by $\mu_1$ the resulting pseudo-distribution. Then, we have $\Abs{\pE_{\mu_1} \left[p \right]} = \frac{\Abs{\pE_\mu \left[p^{2k+1}\right]}}{m^{2k}}$. We are done if this is larger than $\frac{m}{3}$, hence it remains to handle the case $\Abs{\pE_\mu \left[p^{2k+1}\right]}\le \frac{m^{2k+1}}{3}$.

    Now, reweight $\mu$ by $p^{2k-2}$ and denote by $\mu_2$ the resulting pseudo-distribution. Note that by the pseudo-distribution version of Cauchy-Schwarz (\Cref{fact:sos-cauchy-schwarz}), as long as $\mu$ is a degree-$(2k+2)t$ pseudo-distribution,
    \[
        0<\pE_\mu \left[ p^{2k} \right]^2\le \pE_\mu \left[p^{2k-2}\right]\cdot \pE_\mu \left[p^{2k+2}\right]\mcom
    \]
    so that $\pE_\mu \left[p^{2k-2}\right]>0$ and the reweighting is well-defined. Furthermore, we have $\Abs{\pE_{\mu_2}\left[p\right]}=\frac{\Abs{\pE_\mu \left[p^{2k-1}\right]}}{\pE_\mu \left[p^{2k-2}\right]}$. Once again, we are done if this is larger than $\frac{m}{3}$, so we assume from now on that $\Abs{\pE_\mu \left[p^{2k-1}\right]}\le \frac{m}{3}\cdot \pE_\mu \left[p^{2k-2}\right]$.

    Finally, we consider the reweighting of $\mu$ by the SoS polynomial $(p+m)^2 p^{2k-2}$ and call $\mu_3$ the resulting pseudo-distribution. We have:
    \[
        \pE_\mu \left[ (p+m)^2 p^{2k-2} \right]=m^{2k}+2m\cdot \pE_\mu \left[ p^{2k-1} \right]+m^2 \cdot \pE_\mu \left[ p^{2k-2} \right]\in \left(0,\frac{8m^{2k}}{3}\right]\mcom
    \]
    where we also use $\pE_\mu \left[ p^{2k-2} \right]\le m^{2k-2}$ (which follows from \Cref{fact:sos-holders}). In particular, the reweighting for $\mu_3$ is well-defined. Similarly, we have
    \[
        \pE_\mu \left[ (p+m)^2 p^{2k-1}\right]\ge 2m^{2k+1}-\Abs{\pE_\mu \left[ p^{2k+1} \right]}-m^2\cdot \Abs{\pE_\mu \left[ p^{2k-1} \right]}\ge \frac{4m^{2k+1}}{3}\mper
    \]
    Thus, $\pE_{\mu_3} \left[ p \right]\ge \frac{m}{2}$ holds in this case, which concludes the proof.
\end{proof}

We are now ready to prove \Cref{thm:higher-degree-sos}.

\begin{proof}[Proof of \Cref{thm:higher-degree-sos}]
    Similarly to the proof of \Cref{thm:deg6-sos-approx}, we start by defining $q_i = q_i(y,z) = y^\top T_i z$.
    By \Cref{fact:sos-holders} and \Cref{lem:h-v-moment},
    \begin{equation*}
        \SOS^{2k} = \Paren{\pE_{\mu} \iprod{x, q} }^{2k}
        \leq \pE_{\mu} \iprod{x,q}^{2k}
        \leq O\Paren{\frac{n}{k}}^k \pE_{\mu} \E_{h\sim \pmo^n} \iprod{h,q}^{2k} \mper
    \end{equation*}
    Here we require $\mu$ to be a degree-$6k$ pseudo-distribution.

    Since $h\mapsto \pE_{\mu}\iprod{q,h}^{2k}$ is a degree-$2k$ polynomial, by anti-concentration of low-degree polynomials (\Cref{lem:anti-concentration-hypercontractivity}), we can sample $h\in \pmo^n$ such that $\Paren{\pE_{\mu} \iprod{q,h}^{2k}}^{1/2k} \geq \Omega\Bigparen{\sqrt{\frac{k}{n}}} \cdot \SOS$ with probability at least $2^{-O(k)}$.

    The rounding algorithm is as follows,
    \begin{enumerate}
        \item Sample $h\sim \pmo^n$ and set $\ol{x} = h$.

        \item Reweight the pseudo-distribution via \Cref{lem:scalar-fixing-high-degree} such that $|\pE_{\mu'} \iprod{q, h}| \geq \frac{1}{3} \Bigparen{\pE_{\mu} \iprod{q,h}^{2k}}^{1/2k}$.
        The SoS degree required for the reweighting is $2(2k+2) \leq 6k$.

        \item Use Grothendieck rounding (\Cref{fact:grothendieck}) on $\mu'$ to obtain solutions $\ol y,\ol z\in \pmo^n$ for the quadratic polynomial $y^\top (\sum_{i=1}^n h_i T_i) z$.
    \end{enumerate}
    The Grothendieck rounding gives us solutions $\ol y,\ol z \in \pmo^n$ with value $\Omega(1) \cdot |\pE_{\mu'}\iprod{q,h}|$.
    Thus, with probability at least $2^{-O(k)}$, we get assignments $\ol{x}, \ol{y}, \ol{z}\in\pmo^n$ such that
    $f(\ol{x}, \ol{y}, \ol{z}) \geq \Omega\Bigparen{\sqrt{\frac{k}{n}}} \cdot \SOS$.
    This completes the proof.
\end{proof}

\section{Polynomial-size SDPs via compressed SoS relaxations}
\label{sec:compressed}

This section is dedicated to the proof of the following theorem.

\begin{theorem}\label{thm:pruned-version}
    Let $k,n$ be integers such that $1\le k\le n$. There is a $2^{O(k)} n^{O(1)}$-time certification algorithm that, given a decoupled homogeneous degree-3 polynomial $f(x,y,z)=\sum_{1\le i,j,k\le n} T_{ijk} x_i y_i z_k$ achieves $O(\sqrt{n/k})$-approximation to $\OPT\coloneqq \max_{x,y,z\in\pmo^n} f(x,y,z)$.
    Moreover, there is a corresponding  rounding algorithm running in $2^{O(k)} n^{O(1)}$ time that outputs a solution $\ol x,\ol y,\ol z\in\pmo^n$ with value $f(\ol x,\ol y,\ol z)\ge \Omega\Bigparen{\sqrt{\frac{k}{n}}}\cdot \OPT$.
\end{theorem}

Roughly, we will proceed by ``compressing'' the SDP relaxations analyzed in \Cref{sec:rounding-main}. We will use some explicit hitting set of size $2^{k}n^{O(1)}$ and use it to define some  constant-degree SoS relaxations with $2^{k}n^{O(1)}$  variables and one additional axiom.

\subsection{The blockwise construction of the hitting set}
\label{sec:blockwise-precise}

Before explaining how to write down the relaxations, we describe the construction of our hitting set over a small sample space that ``fools'' high moments in every direction. We will mimic the anti-concentration proof from \Cref{sec:approx-higher} by decomposing the $n$-dimensional vectors into $k$ blocks.

\begin{definition} \label{def:hitting-set}
    Let $n, k\in \N$ such that $k$ divides $n$.
    Define the distribution $\calD$ over $\wh x \in \pmo^n$  as follows.
    \begin{enumerate}
        \item Sample $\wh{b}$ from a 4-wise independent distribution over $\pmo^{\frac{n}{k}}$, that is, let $\wh{b} = f(s)$, where $f: \pmo^r \to \pmo^{\frac{n}{k}}$ is a 4-wise independent pseudorandom generator with seed $s \sim \pmo^r$ and $r = O(\log n)$.
        
        \item Sample $\wh{c} \sim \pmo^k$ independently of $\wh{b}$.
        
        \item Let $\wh{x} \coloneqq \wh{c} \ot \wh{b}$. In other words, decompose $\wh{x}$ into $k$ blocks $\wh{x}^{(1)},\dots, \wh{x}^{(k)}$ of size $\frac{n}{k}$ and set $\wh{x}^{(i)} \coloneqq \wh{c}_i \cdot \wh{b}$ for all $i\in [k]$.
    \end{enumerate}
\end{definition}

The following observation can be deduced for example from the classical construction of $k$-wise independent sets of random variables \cite{Jof74}.

\begin{claim}\label{claim:prg-size}
    The distribution $\calD$ can be obtained as the uniform distribution over a sample space of size $2^k n^{O(1)}$. In particular, for any $x\in \supp(\calD)$, $\Pr_{\wh{x} \sim \calD}\left[\wh{x} = x\right] \ge  2^{-k} n^{-O(1)}$.
\end{claim}

We will also need the following result, which is a direct consequence of the Paley-Zygmund inequality and the 4-wise independence of $\wh{b}$.

\begin{claim}\label{lem:paley-zygmund}
    For all $w\in \R^{\frac{n}{k}}$, $\Pr_{\wh{b}} \left[ \Abs{\langle \wh b,w\rangle} \ge \frac{1}{2}\| w\|_2\right]\ge \Omega(1)$.
\end{claim}

Finally, the following lower bound on the moments of $\wh x$ will be the key ingredient to prove that our relaxation provides a correct certificate to the optimum.

\begin{lemma}[Large moments in every direction] \label{lemmom}
    For all $w\in\R^n$,
    \[
        \E_{\wh{x} \sim \calD} \iprod{ \wh{x},w}^{2k} \ge \Omega\left(\frac{k}{n}\right)^k n^{-O(1)}\|w\|_1^{2k} \mper
    \]
\end{lemma}

\begin{proof}
    We first decompose $w$ into $k$ blocks $w^{(1)},\ldots,w^{(k)}$ of size $\frac{n}{k}$, in such a way that $\langle \wh x,w\rangle=\sum_{i=1}^k \wh{c}_i \langle \wh b,w^{(i)}\rangle$. Now for any fixed block $i\in [k]$, we know from \Cref{lem:paley-zygmund} that with at least constant probability over $\wh{b}$, it holds that 
    $|\iprod{\wh{b},w^{(i)}}|\ge \frac{1}{2}\|w^{(i)}\|_2\ge \frac{1}{2}\sqrt{\frac{k}{n}}\|w^{(i)}\|_1$,
    where the last inequality follows from Cauchy-Schwarz. In turn, by linearity of expectation,
    \[
        \E_{\wh{b}}\left[\sum_{i=1}^k \left|\langle \wh{b},w^{(i)}\rangle \right|\right]\ge \Omega(1)\cdot \sqrt{\frac{k}{n}} \| w\|_1\mper
    \]
    In particular, there exists some $x\in\supp(\calD)$ satisfying $|\langle x,w\rangle|\ge \Omega(1)\cdot \sqrt{\frac{k}{n}} \| w\|_1$. By \Cref{claim:prg-size}, this $x$ must be drawn with probability at least $2^{-k}n^{-O(1)}$ from $\calD$. Finally, we apply Markov's inequality to get
    \begin{align*}
        \E_{\wh{x} \sim \calD} \langle \wh{x},w\rangle^{2k}\ge \Omega\left(\frac{k}
        {n}\right)^k\|w\|_1^{2k} \Pr_{\wh{x} \sim \calD}\left[|\langle \wh{x},w\rangle|\ge  \Omega(1)\cdot \sqrt{\frac{k}{n}}\|w\|_1 \right]
        \ge \Omega\left(\frac{k}
        {n}\right)^k n^{-O(1)}\|w\|_1^{2k}\mper
    \end{align*}
    This concludes the proof.
\end{proof}

\subsection{Proof of \texorpdfstring{\Cref{thm:pruned-version}}{Theorem~\ref{thm:pruned-version}}}

We are now ready to state and analyze the SDP relaxation. The high-level intuition is the following: write $q_i\coloneqq\sum_{j,k} T_{ijk} y_j z_k$ for all $i\in [n]$, so that our goal is now to maximize $\langle x,q\rangle$, which by symmetry is equivalent to maximizing $\langle x,q\rangle^2$. Instead of maximizing over $x\in \pmo^n$ we essentially pick a random $\wh x$ from a distribution $\calD$ that has large $2k$-th moments in every direction. Then we replace the objective function $\E_\calD \max_{w}  \langle \wh x,q\rangle^2$  by the following proxy:
\[
    \max_{\mu\text{ pseudo-distribution on $w$}}  \frac{\E_{\wh{x} \sim \calD}\pE_\mu \langle \wh x,q\rangle^{2(k+1)}}{\E_{\wh{x} \sim \calD} \pE_\mu \langle \wh x,q\rangle^{2k}}\mper
\]
As $k$ grows, this yields a sequence of increasingly better approximations leveraging higher moments of the variables. Since expanding the $2k$-th powers would require solving an SDP of size $n^{\Omega(k)}$, we introduce auxiliary variables $\{M_{\wh x}\}$ corresponding to $\langle \wh x,q\rangle^k$ in combinatorial solutions.

\begin{proof}[Proof of \Cref{thm:pruned-version}.]
    Assume without of loss of generality that $k$ divides $n$. Let $\calD$ be the pseudorandom distribution from \Cref{def:hitting-set}. Furthermore, we fix a guess $\alpha\ge 0$ for the value of the optimum of the cubic optimization problem (the final certification and rounding algorithms will be obtained by binary searching for the best possible value of $\alpha$). 
    
    \paragraph{The relaxation.} We solve for \textit{feasibility} the degree-12 SoS program over the following variables:
    \begin{itemize}
        \item variables $y_j$ and $z_k$ for all $j,k\in[n]$. To lighten notations we let $q_i\coloneqq q_i(y,z) = \sum_{1\le j,k\le n} T_{ijk} y_j z_k$ for all $i\in [n]$ (each $q_i$ is a degree-2 polynomial) and write $q = (q_1,\dots,q_n)$.
        
        \item variables $M_{x}$ for each $x\in\supp(\calD)$.

    \end{itemize}
    and under the following additional polynomial constraints:
    \begin{align}
        &y_j^2=1 \quad\text{for all $j\in [n]$}\mcom \nonumber\\
        &z_k^2=1 \quad\text{for all $k\in [n]$}\mcom \nonumber\\
        &\E_{\wh{x} \sim \mathcal D} \left[M_{\wh{x}}^2\left(\langle \wh{x},q\rangle^2-\alpha^2\right)\right] \ge 0\mper\label{eqmain}
    \end{align}

    By construction of $\calD$, the relaxation has $2^k n^{O(1)}$ variables and constraints.

    \paragraph{The rounding algorithm.}
    First, we check that any feasible solution to the SoS program can be rounded into an integral solution $\ol x,\ol y,\ol z\in\pmo^n$ achieving value $\Omega(\alpha)$. Suppose that there exists some degree-12 pseudo-distribution $\mu$ (over $y$, $z$ and $M_x$) satisfying all the constraints. Then by \Cref{eqmain}, there exists $\ol x\in \supp(\calD)$ satisfying
    \[
        \alpha^2\le \frac{\pE_\mu M_{\ol x}^2 \langle \ol x,q\rangle^2}{\pE_\mu M_{\ol x}^2}=\pE_{\mu'} \langle \ol x,q\rangle^2\mcom
    \]
    where $\mu'$ is the degree-6 pseudo-distribution obtained by reweighting $\mu$ by the SoS polynomial $M_{\ol x}^2$. We now use \Cref{lem:scalar-fixing} to construct from $\mu'$ a degree-2 pseudo-distribution $\mu''$ that satisfies
    $\pE_{\mu'} \langle \ol x,q\rangle^2\le 9 \left(\pE_{\mu''}\langle \ol x,q\rangle\right)^2$.

Finally we use Grothendieck rounding (\Cref{fact:grothendieck}) on $\mu''$ to find $\ol y,\ol z\in\pmo^n$ such that 
\[
    f(\ol x,\ol y,\ol z)\ge \frac{1}{K_G} \pE_{\mu''} f(\ol x,y,z)=\frac{1}{K_G}\pE_{\mu''}\langle \ol x,q\rangle\ge \Omega(1)\cdot \alpha\mper
\]

\paragraph{Approximation factor.} Our final algorithm consists of a binary search to get the largest value of $\alpha\ge 0$ that makes the SoS program above feasible. Then, some explicit multiple of $\alpha$ coming from the analysis of our rounding provides a correct upper bound certificate on $\OPT$.

We now check that this achieves approximation   $O(\sqrt{n/k})$. Fix any triplet $x^*,y^*,z^*\in\pmo^n$ and let $q_i^*=\sum_{1\le j,k\le n} T_{ijk} y_j^* z_k^*$ for all $i\in [n]$. Suppose that $(x^*,y^*,z^*)$ achieves the optimum of the original problem, so that $\OPT = \iprod{x^*, q^*} = \|q^*\|_1$. We set $(y,z)=(y^*,z^*)$ and $M_{x}=\langle x,q^*\rangle^k$ for all $x\in\supp(\calD)$, and we prove that this defines a feasible solution. By H\"{o}lder's inequality, $\E_{\wh{x}\sim\calD} \langle \wh{x},q^*\rangle^{2k+2} \geq (\E_{\wh{x}\sim\calD} \langle \wh{x},q^*\rangle^{2k})^{\frac{2k+2}{2k}}$, and \Cref{lemmom} then yields
\[
    \frac{\E_{\wh{x}\sim\calD}M_{\wh x}^2\langle \wh x,q^*\rangle^2}{\E_{\wh{x}\sim\calD} M_{\hat x}^2}
    =\frac{\E_{\wh{x}\sim\calD} \langle \wh{x},q^*\rangle^{2k+2}}{\E_{\wh{x}\sim\calD} \langle \wh{x},q^*\rangle^{2k}}\ge \left(\E_{\wh{x}\sim\calD} \langle \wh x,q^*\rangle^{2k}\right)^{1/k}\ge \Omega(1)\cdot \frac{k}{n}\cdot n^{-O\left(\frac{1}{k}\right)}\OPT^2\mper
\]
Assume without loss of generality that $k=\Omega(\log n)$ (since otherwise, one can always increase $k$ to $\Theta(\log n)$ without affecting the target runtime). Then, as long as $\alpha\le O(1) \cdot\sqrt{\frac{k}{n}} \cdot\OPT$, \Cref{eqmain} is satisfied. This completes the proof.
\end{proof}

\section{Optimization over the unit sphere}
\label{sec:sphere}

In this section, we prove the approximation results for cubic optimization over the unit sphere matching our results over the hypercube:
\begin{itemize}
    \item In \Cref{sec:canonical-sphere}, we show that the canonical degree-$6k$ sum-of-squares relaxation has integrality gap $O(\sqrt{n/k})$ by describing an appropriate rounding algorithm.
    \item In \Cref{sec:opti-sphere}, we prove that a pruned SDP can achieve approximation $O(\sqrt{n/k})$ in time $2^{O(k)}\text{poly}(n)$.
\end{itemize}

\subsection{Analysis of the canonical degree-\texorpdfstring{$k$}{k} SoS relaxation}
\label{sec:canonical-sphere}

We prove that the canonical degree-$k$ SoS relaxation for optimizing over the unit sphere has integrality gap at most $O\Bigparen{\sqrt{\frac{n}{k}}}$. The proof mirrors the hypercube case (\Cref{thm:higher-degree-sos}), although the analysis is much simpler here since we can directly relate the SoS value to the moments of the Gaussian distribution.

\begin{theorem} \label{thm:sos-spherical}
    Fix $1\le k\le n$.
    Given any decoupled homogeneous degree-3   polynomial $f(x,y,z) = \sum_{1\le i,j,k\le n} T_{ijk} x_i y_j z_k$, the canonical degree-$6k$ SoS relaxation of $\max_{x,y,z\in \calS^{n-1}} f(x,y,z)$ has integrality gap $O \Bigparen{\sqrt{\frac{n}{k}}}$.
    Furthermore, given any degree-$6k$ pseudo-distribution $\mu$ over $\left(\calS^{n-1}\right)^3$ such that $\SOS \coloneqq \pE_{\mu} f>0$, 
     there is an $n^{O(k)}$-time randomized rounding algorithm that outputs with high probability $\ol{x},\ol{y},\ol{z}\in \calS^{n-1}$ such that $f(\ol{x}, \ol{y}, \ol{z}) \geq \Omega \Bigparen{\sqrt{\frac{k}{n}}} \cdot \SOS$.
\end{theorem}

\begin{proof}
    Similarly to the proof of \Cref{thm:deg6-sos-approx}, we start by defining $q_i = q_i(y,z) = y^\top T_i z$ and consider
    \[
        \SOS^{2k} = \Paren{\pE_{\mu} \iprod{x, q} }^{2k}
        \leq \pE_{\mu}\iprod{x, q}^{2k}
        \leq \pE_{\mu}\Brac{ \|x\|_2^{2k} \cdot \|q\|_2^{2k}}
        = \pE_{\mu} \|q\|_2^{2k}
        \mcom
    \]
    using Cauchy-Schwarz and the pseudo-expectation version of \Holder's inequality (\Cref{fact:sos-holders}).
    
    Then, let $h\sim \calN(0, I_n)$. We know from standard estimates on the moments of the Gaussian distribution that for any vector $v\in \R^n$, $\E_h \iprod{h, v}^{2k} = c_k \|v\|_2^{2k}$, where $c_k = (2k-1)!! \geq (k/2)^k$.
    Thus, we have
    \begin{equation*}
        \SOS^{2k} \leq \pE_{\mu} \|q\|_2^{2k}
        \leq \Paren{\frac{2}{k}}^k \cdot \E_{h\sim \calN(0,I_n)}\pE_{\mu} \iprod{q, h}^{2k} \mper
    \end{equation*}
    Since $\pE_{\mu}\iprod{q,h}^{2k}$ is a degree-$(2k)$ polynomial in $h$, by \Cref{lem:anti-concentration-hypercontractivity}, $h$ satisfies $\Paren{\pE_{\mu} \iprod{q,h}^{2k}}^{\frac{1}{2k}} \geq \Omega(\sqrt{k}) \cdot \SOS$ with probability at least $2^{-O(k)}$.
    Moreover, with probability at least $1-2^{-\Omega(n)}$ we have $\|h\|_2 \leq O(\sqrt{n})$.

    Hence, our rounding algorithm goes as follows,
    \begin{enumerate}
        \item Sample $h\sim \calN(0, I_n)$ and let $\ol x=\frac{h}{\|h\|_2}$.

        \item Reweight the pseudo-distribution $\mu$ via \Cref{lem:scalar-fixing-high-degree} to obtain a degree-2 pseudo-distribution $\mu'$ such that $\left|\pE_{\mu'} \iprod{q, h}\right| \geq \frac{1}{3} \Bigparen{\pE_{\mu} \iprod{q,h}^{2k}}^{\frac{1}{2k}}$.

        \item Use the lossless rounding for quadratic forms over the sphere (\Cref{lem:rounding-sphere}) on $\mu'$ to obtain solutions $\ol y,\ol z \in \calS^{n-1}$ such that $\ol y^\top (\sum_{i=1}^n h_i T_i) \ol z\ge \left|\pE_{\mu'} \langle q,h\rangle\right|$ (we can flip the sign of $h$ to get the guarantee with the absolute value).
    \end{enumerate}
    Putting everything together, it holds with probability at least $2^{-O(k)}$ that
    \[
        f(\ol x,\ol y,\ol z)=\frac{1}{\|h\|_2}\ol y^\top \left(\sum_{i=1}^n h_i T_i\right) \ol z\ge \frac{1}{\|h\|_2}\left|\pE_{\mu'} \langle q,h\rangle\right|\ge \Omega(1)\cdot \sqrt{\frac{k}{n}}\cdot \SOS \mper
    \]
    Repeating this $\poly(n,2^{k})$ times, we obtain a rounding algorithm that satisfies the desired guarantees with high probablity.
\end{proof}

\subsection{Analysis of compressed SoS relaxations over the unit sphere}
\label{sec:opti-sphere}

In this section, we prove the analogous statement of \Cref{thm:pruned-version} over the sphere.

\begin{theorem}\label{thm:pruned-version-sphere}
    Fix $1\le k\le n$. There is a $2^{O(k)} n^{O(1)}$-time certification algorithm that given a decoupled homogeneous degree-3 polynomial $f(x,y,z)=\sum_{1\le i,j,k\le n} T_{ijk} x_i y_i z_k$ achieves $O(\sqrt{n/k})$-approximation to
    \[
        \OPT\coloneqq \max_{x,y,z\in\calS^{n-1}} f(x,y,z)\mper
    \]
    Moreover, there is a corresponding  rounding algorithm running in $2^{O(k)} n^{O(1)}$ time that outputs a solution $\ol x,\ol y,\ol z\in\calS^{n-1}$ with value $f(\ol x,\ol y,\ol z)\ge \Omega\Bigparen{ \sqrt{\frac{k}{n}}}\cdot \OPT$.
\end{theorem}

Our proof relies on a hitting set construction analogous to \Cref{lemmom}.

\begin{lemma}\label{lem:moments-sphere}
    For any $1\le k\le n$ with $k=\Omega(\log n)$, there exists a distribution $\calD$ over $\calS^{n-1}$ supported on at most $2^{O(k)} n^{O(1)}$ vectors such that for all $w\in\R^n$,
    \[
        \left(\E_{\wh x\sim \mathcal D} \langle \wh x,w\rangle^{2k}\right)^{\frac{1}{2k}}\ge \Omega(1)\cdot \sqrt{\frac{k}{n}} \| w\|_2\mper
    \]
\end{lemma}

\begin{proof}
    Assume without loss of generality that $k$ divides $n$. We define the distribution $\calD$ over $\wh x \in \calS^{n-1}$ as follows.
    \begin{enumerate}
        \item Sample $\wh b\in\pmo^{\frac{n}{k}}$ from a 4-wise independent distribution. Similarly to \Cref{claim:prg-size}, this can be achieved by taking the uniform distribution over a subset of $\pmo^{\frac{n}{k}}$  of size $n^{O(1)}$.
        \item Sample independently $\wh c$ uniformly over an $\varepsilon$-net of $\calS^{k-1}$ of size $O(1/\varepsilon)^k$ for $\varepsilon=\frac{1}{10}$.
        \item Output $\wh x=\sqrt{\frac{k}{n}}\cdot \wh c\otimes \wh b$. Note in particular that $\wh x$ is a unit vector.
    \end{enumerate}
    Decompose $w$ in $k$ blocks $w^{(1)},\ldots,w^{(k)}$ of size $\frac{n}{k}$ such that $\langle \wh x,w\rangle=\sqrt{\frac{k}{n}}\cdot\sum_i \wh{c}_i \langle \wh{b},w^{(i)} \rangle$.
    \Cref{lem:paley-zygmund} states that for each $1 \leq i \leq \frac{n}{k}$,
    \[
        \Pr_{\wh{b}} \Brac{ \Abs{\iprod{ \wh{b},w^{(i)}}} \ge \frac{1}{2}\|w^{(i)}\|_2 }\ge \Omega(1)\mper
    \]
    This implies that $\E_{\wh{b}} \sum_{i=1}^k \langle \wh b,w^{(i)}\rangle^2\ge \Omega(1)\cdot \|w\|_2^2$. On the other hand, for any $b\in\pmo^{\frac{n}{k}}$, we can find $c=c(b)$ in the $\varepsilon$-net such that
    \[
        \sum_{i=1}^k c_i \langle b,w^{(i)} \rangle\ge \Omega(1)\cdot \left(\sum_{i=1}^k \langle b, w^{(i)}\rangle^2\right)^{\frac{1}{2}}\mper
    \]
    Therefore, there must exist $x\in \supp(\calD)$ such that $\langle x,w\rangle\ge \Omega(1)\cdot\sqrt{\frac{k}{n}}\cdot \|w\|_2$, and this $x$ is drawn with probability at least $2^{-O(k)} n^{-O(1)}$. Finally, by Markov inequality,
    \begin{align*}
        \E_{\wh{x}\sim\calD} \langle \wh x,w\rangle^{2k}&\ge \Pr_{\wh{x}\sim\calD}\left(|\langle \wh x,w\rangle|
        \ge \Omega(1)\cdot\sqrt{\frac{k}{n}}\cdot \|w\|_2\right)\cdot \Omega\left(\frac{k}{n}\right)^k\|w\|_2^{2k}\\&= 2^{-O(k)} n^{-O(1)} \Omega\left(\frac{k}{n}\right)^k \|w\|_2^{2k}\mper
    \end{align*}
    This completes the proof assuming $k=\Omega(\log n)$.
\end{proof}

We are now ready to prove \Cref{thm:pruned-version-sphere}.

\begin{proof}[Proof of \Cref{thm:pruned-version-sphere}] The proof is identical to the proof of \Cref{thm:pruned-version}, with the following exceptions:
    \begin{itemize}
        \item the boolean constraints $y_j^2=1$ and $z_k^2=1$ for all $j,k\in[n]$ become spherical constraints: $\|y\|_2^2=1$ and $\|z\|_2^2=1$.
        
        \item instead of Grothendieck's rounding we use the (lossless) rounding for quadratic forms over the sphere from \Cref{lem:rounding-sphere}.
        
        \item if $x^*, y^*,z^*$ achieve the optimum for the cubic maximization problem, then $\OPT=\|q^*\|_2$, where $q_i^*=\sum_{1\le j,k\le n} T_{ijk} y_j^* z_k^*$ for all $i\in [n]$.

        \item the last sequence of inequalities in the proof of \Cref{thm:pruned-version} becomes as follows after using \Cref{lem:moments-sphere} instead of \Cref{lemmom}:
            \[
                \frac{\E_{\calD}\left[M_{\wh x}^2\langle \wh x,q^*\rangle^2\right]}{\E_{\calD} \left[M_{\hat x}^2\right]}=\frac{\E_{\calD} \langle \wh{x},q^*\rangle^{2k+2}}{\E_{\calD} \langle \wh{x},q^*\rangle^{2k}}\ge \left(\E_{\calD} \langle \wh x,q^*\rangle^{2k}\right)^{1/k}\ge \Omega(1)\cdot \frac{k}{n}\cdot n^{-O\left(\frac{1}{k}\right)}\OPT^2\mper
            \]
    \end{itemize}
    We conclude in the same way as in the proof of \Cref{thm:pruned-version}.
\end{proof}

\section{Rounding algorithm for 3SAT}
\label{sec:3-sat}

In this section, we consider 3SAT formulas where each 3-tuple of variables appears at most once, i.e.\ there are no two clauses with the same set of variables.
\Hastad and Venkatesh~\cite{HV04} used an anti-concentration result of \cite{AGK04} to prove that any 3SAT formula with $m$ clauses has value at least $\frac{7}{8} + \Omega(\frac{1}{\sqrt{m}})$ (which is achieved by a random assignment with probability $\Omega(\frac 1 m)$).

We prove the following improvement over this result:

\begin{theorem}[Restatement of \Cref{thm:3sat}]
    \label{thm:3sat-improved}
    There is a polynomial-time randomized algorithm that, given a satisfiable 3SAT formula with $n$ variables, finds with high probability an assignment satisfying a $(\frac 7 8 + \wt{\Omega}(n^{-\frac 3 4}))$-fraction of the clauses.
\end{theorem}

\paragraph{Notations.} A 3SAT clause $C$ with variables $x_1,x_2,x_3$ and negation pattern $(\sigma_1,\sigma_2,\sigma_3)\in \pmo^3$ can be written as
\begin{equation*}
    \psi_C(x_1,x_2,x_3) = \frac{7}{8} - \frac{1}{8}(\sigma_1 x_1 + \sigma_2 x_2 + \sigma_3 x_3 + \sigma_1 \sigma_2 x_1 x_2 + \sigma_2 \sigma_3 x_2 x_3 + \sigma_1 \sigma_3 x_1 x_3 + \sigma_1 \sigma_2 \sigma_3 x_1x_2x_3) \mper
\end{equation*}
Here we adopt the convention that $-1$ is $\mathsf{True}$ and $+1$ is $\mathsf{False}$, and we can see that $\psi_C(x) = 0$ if $\sigma_1 x_1 = \sigma_2 x_2 = \sigma_3 x_3 = +1$, and $\psi_C(x) = 1$ otherwise.
Thus, a 3SAT formula can be represented as a function $\psi: \pmo^n \to [0,1]$,
\begin{equation*}
    \psi(x) = \frac{7}{8} + f_1(x) + f_2(x) + f_3(x) \mcom
\end{equation*}
where $f_1, f_2, f_3$ are homogeneous polynomials of degree 1, 2 and 3 respectively. 

Observe that $\max_{x\in \pmo^n} |f_1(x)|,\max_{x\in \pmo^n}|f_2(x)|\le 3/8$, and $\max_{x\in \pmo^n} |f_3(x)| \leq 1/8$. In particular, this last statement implies the following crucial observation:

\begin{observation}
\label{obs:satisfiable-3sat}
If $x^*$ is a satisfying assignment, then $f_1(x^*) + f_2(x^*) \geq 0$.
\end{observation}

Before proceeding to describing the algorithm, we show the following version of degree-3 decoupling that augments the guarantees of \Cref{lem:decoupling} by controlling also the degree-1 and degree-2 parts:
\begin{lemma}[Recoupling]
    \label{lem:recoupling-improved}

    Given $x,y,z \in \pmo^n$, there exists a polynomial-time sampleable distribution $\calD$ over $\pmo^n$ such that
    \begin{enumerate}
        \item For any degree-3 homogeneous multilinear polynomial $f(x)=\sum_{i,j,k\in [n]} T_{ijk} x_i x_j x_k$ (where $T$ is a symmetric $3$-tensor) and the corresponding decoupled polynomial $\wt{f}(x,y,z)=\sum_{i,j,k\in [n]} T_{ijk} x_i y_j z_k$, 
        \[
            \E_{x'\sim \calD}[f(x')] = \frac{2}{9} \cdot \wt{f}(x,y,z)\mper
        \]

        \item For any degree-2 homogeneous multilinear polynomial $g(x)=\sum_{i,j\in [n]} M_{ij} x_i x_j$,
        \[
            \E_{x'\sim \calD} [g(x')] = \frac{1}{9} \cdot (g(x) + g(y) + g(z))\mper
        \]

        \item $\E_{x'\sim \mathcal D} \left[x'\right]=0$.
    \end{enumerate}
\end{lemma}
\begin{proof}
    The distribution $\calD$ can be sampled as follows:
    \begin{itemize}
        \item Sample $b_1$ and $b_2$ independently and uniformly in $\pm 1$ and let $b_3=b_1b_2$. Then $(b_1, b_2, b_3)$ has a pairwise independent distribution and $b_1 b_2 b_3 = 1$.

        \item Independently for each $i\in [n]$, sample $x_i'$ uniformly in the multiset $\{b_1 x_i,b_2 y_i,b_3 z_i\}$.
    \end{itemize}
    Then, we have:
    \begin{equation*}
    \begin{aligned}
        \E_{x'\sim \calD} \left[f(x')\right]
        &= \sum_{1\le i,j,k\le n} T_{ijk} \cdot \E \left[ \frac{b_1 x_i + b_2 y_i + b_3 z_i}{3}\cdot \frac{b_1 x_j + b_2 y_j + b_3 z_j}{3} \cdot \frac{b_1 x_k + b_2 y_k + b_3 z_k}{3} \right] \\
        &= \frac{1}{27} \sum_{1\le i,j,k\le n} T_{ijk} \cdot (x_i y_j z_k + x_i z_j y_k + x_j y_i z_k + x_j y_k z_i + x_k y_j z_i + x_k y_i z_j) \\
        &= \frac{2}{9} \cdot \wt{f}(x,y,z) \mper
    \end{aligned}
    \end{equation*}

    Similarly, for degree-2 homogeneous multilinear polynomials,
    \begin{align*}
        \E_{x'\sim \calD} \left[g(x')\right] &= \sum_{1\le i,j\le n} M_{ij} \cdot \E \left[ \frac{b_1 x_i + b_2 y_i + b_3 z_i}{3}\cdot \frac{b_1 x_j + b_2 y_j + b_3 z_j}{3}\right] 
        = \frac{1}{9} \cdot (g(x) + g(y) + g(z)) \mper
    \end{align*}
    The third part follows similarly.
\end{proof}

\paragraph{The algorithm.} We now describe the algorithm that achieves the guarantees of \Cref{thm:3sat-improved}. We define $\delta\coloneqq\frac{c}{\sqrt{n}\log n}$ (for some small constant $c>0$ to be picked at the end of \Cref{lem:large-deg3}).

\begin{mdframed}
    \begin{algorithm}[3SAT]
    \label{alg:3sat-improved}
    \mbox{}
      \begin{description}
      \item[Input:] A 3SAT instance in $n$ variables: $\psi(x) = \frac{7}{8} + f_1(x) + f_2(x) + f_3(x)$ where $f_1, f_2, f_3$ are homogeneous polynomials of degree 1, 2 and 3 respectively.

      \item[Output:] An assignment $x\in \pmo^n$ with value $\frac{7}{8}+\wt{\Omega}(n^{-\frac 3 4})$ with high probability.

      \item[Operation:] \mbox{}
        \begin{enumerate}
     
            \item Solve for feasibility the following two degree-2 SoS programs over variables $x=(x_1,\ldots,x_n)$:
            \begin{enumerate}
                \item with axioms $x_i^2=1$ for all $i\in [n]$, $f_1(x)+f_2(x)\ge 0$, and $f_1(x)>\delta$;
                \item with axioms $x_i^2=1$ for all $i\in [n]$, $f_1(x)+f_2(x)\ge 0$, and $f_1(x)<-\delta$.
            \end{enumerate}
            If none of these programs is feasible, move to the next case.

            Otherwise, let $\mu$ be either a feasible degree-2 pseudo-distribution for (a), or the negation of a feasible degree-2 pseudo-distribution for (b).
            \begin{itemize}
                \item Sample $g\sim \mathcal N(\pE_\mu x,\pE_\mu (x-\pE_\mu x)(x-\pE_\mu x)^\top)$. 
                \item For all $i\in [n]$, set $\ol x_i=\frac{g_i}{T}$ if $|g_i|\le T$ and $\ol x_i=\text{sign}(g_i)$ otherwise (for some parameter $T=T(n)>0$ to be chosen in \Cref{lem:large-deg1}).
                \item Sample $x^{(1)}_i=1$ with probability $\frac{1+p\ol x_i}{2}$ and $x^{(1)}_i=-1$ with probability $\frac{1-p\ol x_i}{2}$, independently for all $i\in [n]$ (for some parameter $p=p(n)>0$ to be chosen in \Cref{lem:large-deg1}).
            \end{itemize}

            \item Get $x'$ by optimizing $f_2$ using the second part of \Cref{thm:charikar-wirth}.
            Then, set $x^{(2)} = x'$ or $x^{(2)}=-x'$ depending on which of the two has higher value.

            \item Let $\wt{f_3}(x,y,z)$ be the decoupled polynomial of $f_3(x)$.
            \begin{itemize}
                \item Run the degree-6 SoS relaxation of cubic optimization of $\wt f_3$ in variables $x,y,z\in \pmo^n$ with the additional axioms $f_2(y) \geq -\delta$ and $f_2(z) \geq -\delta$.

                \item Sample $\ol x\sim \pmo^n$ and reweight the pseudo-distribution $\mu$ as in the algorithm of \Cref{thm:deg6-sos-approx}.

                \item Apply Charikar-Wirth rounding (\Cref{thm:charikar-wirth}) on the reweighted $\mu'$ to get $\ol y,\ol z\in \pmo^n$.

                \item Finally, obtain $x^{(3)}$ by recoupling $\ol x,\ol y,\ol z$ via \Cref{lem:recoupling-improved}.

            \end{itemize}

            \item Pick the best assignment among $x^{(1)}, x^{(2)}, x^{(3)}$.

            \item Repeat the steps 1 to 4 $\text{poly}(n)$ times and output the best assignment obtained.
        \end{enumerate}
      \end{description}
    \end{algorithm}
\end{mdframed}

The first three steps in the algorithm correspond to each of the degree-1, 2, or 3 part being large. On a high level, our strategy is quite natural. If the degree-2 part is large, we use the classical roundings for degree-2 polynomials. If the degree-3 part is large, then we use our rounding for decoupled degree-3 polynomials from \Cref{sec:rounding-main}. In those two cases, we get an assignment with value $\frac 7 8 + \wt \Omega(n^{-\frac 1 2})$. If the degree-1 part is large, we introduce a different algorithm based on a degree-2 SoS relaxation with additional axioms. Our rounding in this case is inspired by the proof of \Cref{thm:charikar-wirth} and uses an additional idea to make the degree-3 part negligible. In this last case, we get an assignment with value $\frac 7 8 + \wt \Omega(n^{-\frac 3 4})$.

We analyze separately these three cases now.

\begin{lemma}[Large degree-1 part]
    \label{lem:large-deg1}
    Suppose that $\mu$ is a degree-2 pseudo-distribution such that:
    \[
        \left|\pE_\mu \left[f_1(x)\right]\right|\ge \delta\mcom\quad \pE_\mu \left[f_1(x)+f_2(x)\right]\ge 0\mper
    \]
    Then $x^{(1)}$ has value $\frac{7}{8}+\wt{\Omega}(n^{-\frac 3 4})$.
\end{lemma}

\begin{proof}
    First, up to negating the pseudo-distribution (which does not affect the degree-2 part), we can assume without loss of generality that $\pE_\mu \left[f_1(x)\right]\ge \delta$.
    
    The rounding proceeds in two steps. We introduce some parameters $p=p(n)\in [0,1]$ and $T=T(n)>0$ to be fixed later.
    \begin{enumerate}
        \item Sample a Gaussian with mean and covariance given by the degree-2 pseudo-distribution: $g\sim \mathcal N(\pE_\mu x,\pE_\mu (x-\pE_\mu x)(x-\pE_\mu x)^\top)$. Then, let $\ol x\in [-1,1]^n$ be defined as follows: for each $i\in [n]$, let $\ol x_i=\frac{g_i}{T}$ if $|g_i|\le T$ and $\ol x_i=\text{sign}(g_i)$ otherwise.
        \item Use $\ol x$ as a bias for sampling $ x^{(1)}$: sample for each $i\in [n]$ independently: $x_i^{(1)}=1$ with probability $\frac{1+p\ol x_i}{2}$ and $x^{(1)}_i=-1$ with probability $\frac{1-p\ol x_i}{2}$.
    \end{enumerate}
    
    We now analyze this rounding. Define
    \begin{align*}
        \Delta_i\coloneqq\E \ol x_i-\frac{1}{T}\E g_i, \quad \Delta_{ij}\coloneqq\E\left[ \ol x_i \ol x_j\right] - \frac{1}{T^2}\E\left[g_i g_j\right]\mper
    \end{align*}

    Then, at the end of the first step, we have (here, $\|f_1\|_1$ and $\|f_2\|_1$ denote the sum of the absolute value of the coefficients of $f_1$  and $f_2$, respectively):
    \begin{align}
        \E \left[f_1(\ol x)\right] &\ge \frac{1}{T}\pE\left[f_1(x)\right] - \|f_1\|_1 \cdot \max_{1\le i\le n}  |\Delta_i|\mcom\label{eq:deg1}\\
        \E \left[f_2(\ol x)\right] &\ge \frac{1}{T^2}\pE\left[f_2(x)\right] - \|f_2\|_1 \cdot \max_{1\le i,j\le n}  |\Delta_{ij}|\mper
    \end{align}

    \begin{claim}
        $\max_{1\le i\le n} |\Delta_i|$ and $\max_{1\le i,j\le n} |\Delta_{ij}|$ are both at most $O(1)\cdot e^{-\frac{T^2}{8}}$.
    \end{claim}
    
    \begin{proof}
        Fix $i\in [n]$. First,
        \[
            \Delta_i=\frac{1}{T}\E \left[g_i (\mathbf 1_{|g_i|\le T}-1)\right]+\E\left[\text{sign}(g_i) \mathbf 1_{|g_i|>T}\right]\mper
        \]
        Since $\mu$ is a degree-2 pseudo-distribution over the hypercube, $g_i$ is a Gaussian with mean $\pE_\mu x_i\in [-1,1]$ and variance 1. Define $p_i\coloneqq \Pr(|g_i|>T)\le e^{-T^2/4}$ (this holds provided $T$ is a large enough constant). By the triangle inequality and Cauchy-Schwarz, we have
        \[
            |\Delta_i|\le \frac{1}{T}\left|\E \left[g_i \mathbf 1_{|g_i|> T}\right]\right|+p_i\le \frac 1 T \sqrt{2p_i}+p_i\le O(1)\cdot e^{-\frac{T^2}{8}}\mper
        \]
        Fix now $i,j\in [n]$. Similarly,
        \begin{align*}
            \Delta_{ij}&=\frac{1}{T^2}\E\left[g_i g_j\left(\mathbf 1_{|g_i|\le T,|g_j|\le T}-1\right)\right]+\frac 1 T \E\left[\text{sign}(g_i) g_j\mathbf 1_{|g_i|>T,|g_j|\le T}\right]\\
            &+\frac 1 T  \E\left[g_i \text{sign}(g_j)\mathbf 1_{|g_i|\le T,|g_j|>T}\right] + \E\left[\text{sign}(g_i) \text{sign}(g_j)\mathbf 1_{|g_i|> T,|g_j|>T}\right]\mper
        \end{align*}
        Note that $(g_i, g_j)$ is a 2-dimensional Gaussian with marginal means bounded by 1 in absolute value and $|\E[g_i g_j]|=|\pE_\mu [x_i x_j]|\le 1$ by \Cref{fact:sos-cauchy-schwarz}. Hence, using once again the triangle inequality and Cauchy-Schwarz,
        \begin{align*}
            |\Delta_{ij}|&\le \frac{1}{T^2}\left|\E [g_i g_j \mathbf 1_{|g_i|>T\text{ or } |g_j|>T}]\right| +\frac{1}{T}\left(\E\left[|g_j|\mathbf 1_{|g_i|>T}\right]+\E\left[|g_i|\mathbf 1_{|g_j|>T}\right]\right)+p_i\\
            &\le \frac{1}{T^2}\sqrt{\E \left[g_i^2 g_j^2\right]\left(p_i+p_j\right)} + \frac{\sqrt 2}{T}\left(\sqrt{ p_i}+\sqrt{ p_j}\right)+p_i\\
            &\le O(1)\cdot e^{-\frac{T^2}{8}}\mper\qedhere
        \end{align*}
    \end{proof}
    Since $f_1$ and $f_2$ come from a 3SAT instance with $m\le n^3$ clauses, we have $\|f_1\|_1,\|f_2\|_1\le O(n^3)$. Thus, if we pick $T=\sqrt{48\log n}$, we get
    \begin{align*}
        \E \left[f_1(\ol x)+f_2(\ol x)\right]&\ge \frac{1}{\sqrt{48\log n}}\pE\left[f_1(x)\right]+\frac{1}{48\log n}\pE\left[f_2(x)\right]-O\left(\frac{1}{n^{2}}\right)\ge -O\left(\frac{1}{n^{2}}\right)\mcom  
    \end{align*}
    where we used our assumption $\pE_\mu \left[ f_1(x)+f_2(x) \right]\ge 0$ (together with the fact that $\pE\left[f_1(x)\right]\ge 0)$.
    Next, at the end of second step, it holds that:
    \begin{align}
        \E \left[f_1( x^{(1)})\right]&= p \E \left[f_1(\ol x)\right],\quad\label{eq:deg12}
        \E \left[f_2(x^{(1)})\right] = p^2 \E \left[f_2(\ol x)\right]\mcom\\
        \E \left[f_3(x^{(1)})\right] &= p^3 \E \left[f_3(\ol x)\right]\ge -p^3\mper\label{eq:deg3negbound}
    \end{align}
    Hence, by \Cref{eq:deg12} and \Cref{eq:deg1}:
    \begin{align*}
        \E \left[f_1( x^{(1)})+f_2(x^{(1)})\right]&\ge (p-p^2)\E \left[f_1(\ol x)\right]+p^2 \E \left[f_1(\ol x)+f_2(\ol x)\right]\\
        &\ge\Omega(1)\cdot \frac{\delta(p-p^2)}{\sqrt{\log n}}- O\left(\frac{p^2}{n^{2}}\right)-O\left(\frac p {n^2}\right)\mper
    \end{align*}
    Setting $p=\frac{n^{-\frac 1 4}}{\log n}$, we get
    \[
        \E \left[f_1( x^{(1)})+f_2( x^{(1)})\right]\ge \Omega(1)\cdot \frac{n^{-\frac{3}{4}}}{\log^{2.5} n}\mcom
    \]
    and we conclude by combining with \Cref{eq:deg3negbound}.
\end{proof}

\begin{lemma}[Large degree-2 part]
    \label{lem:large-deg2}
    Suppose that $\max f_2(x)\ge \delta$. Then $x^{(2)}$ has value $\frac{7}{8}+\wt{\Omega}(n^{-\frac 1 2})$.
\end{lemma}

\begin{proof}
    Using Charikar-Wirth rounding (\Cref{thm:charikar-wirth} with $T=\Theta(\sqrt{\log n})$), we get that 
    \[
        f_2(x')\ge \Omega\left(\frac{1}{\log n}\right)\max_{x\in \pmo^n} f_2(x)\mper
    \]
    Moreover, $f_2$ is invariant by the choice of the signing $\pm x'$, while $f_1+f_3$ changes sign. Thus, $x^{(2)}$ has value $\frac{7}{8}+\wt{\Omega}(n^{-\frac 1 2})$ (recall that $\delta=\frac{c}{\sqrt{n} \log n}$).
\end{proof}

\begin{lemma}[Large degree-3 part]
    \label{lem:large-deg3}
    Suppose that $\max_x f_2(x)\le \delta$ and $|f_1(x^*)|\le \delta$ for some assignment $x^*$ satisfying the 3SAT formula. Then $x^{(3)}$ has value $\frac 7 8 + \wt{\Omega}(n^{-\frac 1 2})$.
\end{lemma}

\begin{proof}
        We run the canonical degree-6 SoS relaxation on the decoupled polynomial $\wt{f_3}(x,y,z)$ associated to $f_3$ with variables $x,y,z\in \pmo^n$, and the additional axioms $f_2(y)\ge -\delta$ and $f_2(z)\ge -\delta$. Let $\mu$ denote the resulting pseudo-distribution.

        Since $\max_x f_2(x)\le \delta$, $|f_1(x^*)|\le \delta$ and $\delta=o(1)$, we must have $f_3(x^*)\ge \frac{1}{8}-o(1)$. Using \Cref{obs:satisfiable-3sat}, we get that the delta pseudo-distribution centered at $x=y=z=x^*$ is a feasible solution to the SoS relaxation, so that $\mu$ satisfies 
        \begin{equation*}
            \pE_{\mu} \wt{f_3}(x,y,z) \geq \frac{1}{8} - o(1),\quad 
            \pE_{\mu} f_2(y) \geq -\delta, \quad
            \pE_{\mu} f_2(z) \geq -\delta\mper
        \end{equation*}

        Next, we sample $\ol x  \sim \pmo^n$ and reweight $\mu$ using \Cref{lem:scalar-fixing} to get a degree-2 pseudo-distribution $\mu'$.
        The same analysis shows that with at least constant probability, $\ol x$ satisfies
        \begin{equation}
            \label{eq:goodxsat}
            \pE_{\mu'} \wt{f}_3(\ol x,y,z) \geq \Omega\Paren{\frac{1}{\sqrt{n}}} \mper
        \end{equation}
        Furthermore, we claim that for $C > 1$,
        \begin{equation*}
            \Pr_{\ol x\sim\pmo^n}\Brac{f_2(\ol x) \geq - C\delta} \ge 1-\frac{1}{C} \mper
        \end{equation*}
        To see this, note that $\E_{\ol x} f_2(\ol x) = 0$ because $f_2$ is multilinear.
        The above bound then follows from the assumption $\max_x f_2(x) \leq \delta$.
        Therefore, by a union bound, with at least constant probability we get a good $\ol x$ that satisfies simultaneously \Cref{eq:goodxsat} and $f_2(\ol x) \geq -C\delta$. By repeating the sampling  $\text{poly}(n)$ times, we can find such an $\ol x\in \pmo^n$ with high probability.

        Now, fix $\ol x\in \pmo^n$ satisfying the previous conditions. We apply \Cref{thm:charikar-wirth} to the degree-2 pseudo-distribution $\mu'$ over $(y,z)\in \pmo^{2n}$ for some $T>0$ to be fixed later. Denoting by $\|f_2\|_1$ (resp. $\|f_3\|_1$) the sum of the absolute value of the coefficients of $f_2$ (resp. $f_3$), we have:
        \begin{equation*}
        \begin{aligned}
            \E_{\ol y} \left[f_2(\ol y)\right] &\ge \frac{1}{T^2}\pE_{\mu'} \left[f_2(y)\right]-8e^{-\frac{T^2}{2}}\|f_2\|_1 \mcom \\
            \E_{\ol z} \left[f_2(\ol z)\right] &\ge  \frac{1}{T^2}\pE_{\mu'} \left[f_2(z)\right]-8e^{-\frac{T^2}{2}}\|f_2\|_1 \mcom \\
            \E_{\ol y,\ol z} \left[\wt{f_3}(\ol x,\ol y,\ol z)\right] &\ge  \frac{1}{T^2}\pE_{\mu'} \left[\wt{f_3}(\ol x,y,z)\right]-8e^{-\frac{T^2}{2}}\|f_3\|_1 \mper \\
        \end{aligned}
        \end{equation*}
        The constraints $f_2(y)\ge -\delta$ and $f_2(z)\ge -\delta$ still hold for the reweighted pseudo-distribution $\mu'$. Moreover, $\|f_2\|_1$ and $\|f_3\|_1$ are both $O(n^3)$, so by picking $T=\sqrt{8\log n}$, we get
        \begin{equation*}
        \begin{aligned}
            \E_{\ol y} \left[f_2(\ol y)\right] &\ge -O\left(\frac{\delta}{\log n}\right)-O\left(\frac{1}{n}\right) \mcom \\
            \E_{\ol z} \left[f_2(\ol z)\right] &\ge -O\left(\frac{\delta}{\log n}\right)-O\left(\frac{1}{n}\right) \mcom \\
            \E_{\ol y,\ol z} \left[\wt{f_3}(\ol x,\ol y,\ol z)\right] &\ge  \Omega\left(\frac{1}{\sqrt n \log n}\right)-O\left(\frac{1}{n}\right) \mper \\
        \end{aligned}
        \end{equation*}

        Finally, once we have $\ol x,\ol y,\ol z$, we recouple them to get $x^{(3)}$ by \Cref{lem:recoupling-improved}, obtaining
        \begin{equation*}
        \begin{aligned}
            \E_{x^{(3)}} f_3(x^{(3)}) &= \frac{2}{9} \cdot \wt{f_3}(\ol x,\ol y,\ol z)\mcom\\
            \E_{x^{(3)}} f_2(x^{(3)}) &= \frac{1}{9}\cdot (f_2(\ol x) + f_2(\ol y) + f_2(\ol z))\mcom\\
            \E_{x^{(3)}} f_1(x^{(3)}) &= 0
            \mper
        \end{aligned}
        \end{equation*}
        Thus, we have
        \begin{equation*}
            \E \psi(x^{(3)}) \geq \frac{7}{8} + \Omega\Paren{\frac{1}{\sqrt{n} \log n}} - O(\delta)
            = \frac{7}{8} + \wt{\Omega}\Paren{\frac{1}{\sqrt{n}}}\mcom
        \end{equation*}
        where the last equality holds provided that we pick the constant $c>0$ in the definition of $\delta$ to be small enough. This concludes the proof.
\end{proof}

\begin{proof}[Proof of \Cref{thm:3sat-improved}]

    It now suffices to prove that one of the assumptions of \Cref{lem:large-deg1}, \Cref{lem:large-deg2} or \Cref{lem:large-deg3} must hold. Fix some satisfying assignment $x^*$ to the 3SAT formula.

    \begin{enumerate}
        \item If $|f_1(x^*)|>\delta$, then one of the two SoS programs from Step 1 of \Cref{alg:3sat-improved} is feasible, and the assumptions of \Cref{lem:large-deg1} hold.

        \item If $\max_x f_2(x)> \delta$, then the assumptions of \Cref{lem:large-deg2} hold.

        \item If $|f_1(x^*)|\le \delta$ and  $\max_x f_2(x)\le \delta$, then the assumptions of \Cref{lem:large-deg3} hold.
    \end{enumerate}
    Hence, in all cases we get a random assignment $\wt x\in \pmo^n$ satisfying $\E \psi(\wt x)\ge \frac{7}{8}+\wt{\Omega}(n^{-\frac 3 4})$. By repeating the rounding $\text{poly}(n)$ times, we can get such an assignment with high probability.
\end{proof}

\section*{Acknowledgements}

We would like to thank anonymous reviewers for their valuable feedback.
We would also like to thank Jonathan Shi and Antares Chen for inspiring conversations on related polynomial optimization problems.
Jun-Ting Hsieh also thanks Prashanti Anderson for discussions on decoupling inequalities.
Finally, part of this work was done while Jun-Ting Hsieh was visiting Bocconi University in 2022.

\bibliographystyle{alpha}
\bibliography{main}

\newcommand{\etalchar}[1]{$^{#1}$}
\begin{thebibliography}{BGG{\etalchar{+}}17}

\bibitem[AGK04]{AGK04}
Noga Alon, Gregory Gutin, and Michael Krivelevich.
\newblock Algorithms with large domination ratio.
\newblock {\em Journal of Algorithms}, 50(1):118--131, 2004.

\bibitem[AIM14]{AIM14}
Scott Aaronson, Russell Impagliazzo, and Dana Moshkovitz.
\newblock {AM with multiple Merlins}.
\newblock In {\em {IEEE} 29th Conference on Computational Complexity, {CCC}
  2014, Vancouver, BC, Canada, June 11-13, 2014}, pages 44--55. IEEE, 2014.

\bibitem[AN06]{AN06}
Noga Alon and Assaf Naor.
\newblock {Approximating the Cut-Norm via Grothendieck's Inequality}.
\newblock {\em SIAM Journal on Computing}, 35(4):787--803, 2006.

\bibitem[BBH{\etalchar{+}}12]{BBH+12}
Boaz Barak, Fernando~G.S.L. Brandao, Aram~W. Harrow, Jonathan Kelner, David
  Steurer, and Yuan Zhou.
\newblock {Hypercontractivity, sum-of-squares proofs, and their applications}.
\newblock In {\em Proceedings of the 44th Symposium on Theory of Computing
  Conference, {STOC} 2012, New York, NY, USA, May 19 - 22, 2012}, pages
  307--326, 2012.

\bibitem[BGG{\etalchar{+}}17]{BGG+17}
Vijay Bhattiprolu, Mrinalkanti Ghosh, Venkatesan Guruswami, Euiwoong Lee, and
  Madhur Tulsiani.
\newblock {Weak Decoupling, Polynomial Folds and Approximate Optimization over
  the Sphere}.
\newblock In {\em 58th {IEEE} Annual Symposium on Foundations of Computer
  Science, {FOCS} 2017, Berkeley, CA, USA, October 15-17, 2017}, pages
  1008--1019. IEEE, 2017.

\bibitem[BKS17]{BKS17}
Boaz Barak, Pravesh~K. Kothari, and David Steurer.
\newblock {Quantum entanglement, sum of squares, and the log rank conjecture}.
\newblock In {\em Proceedings of the 49th Annual {ACM} {SIGACT} Symposium on
  Theory of Computing, {STOC} 2017, Montreal, QC, Canada, June 19-23, 2017},
  2017.

\bibitem[Bog98]{Bog98}
Vladimir~I. Bogachev.
\newblock {\em {Gaussian Measures}}, volume~62 of {\em Mathematical Surveys and
  Monographs}.
\newblock American Mathematical Society, Providence, RI, 1998.

\bibitem[BRS11]{BRS11}
Boaz Barak, Prasad Raghavendra, and David Steurer.
\newblock {Rounding Semidefinite Programming Hierarchies via Global
  Correlation}.
\newblock In {\em {IEEE} 52nd Annual Symposium on Foundations of Computer
  Science, {FOCS} 2011, Palm Springs, CA, USA, October 22-25, 2011}, pages
  472--481. IEEE, 2011.

\bibitem[BS16]{BS16}
Boaz Barak and David Steurer.
\newblock {Proofs, beliefs, and algorithms through the lens of sum-of-squares}.
\newblock {\em Course notes:
  \url{http://www.sumofsquares.org/public/index.html}}, 2016.

\bibitem[CW04]{CW04}
Moses Charikar and Anthony Wirth.
\newblock {Maximizing Quadratic Programs: Extending Grothendieck's Inequality}.
\newblock In {\em 45th Symposium on Foundations of Computer Science {(FOCS}
  2004), 17-19 October 2004, Rome, Italy, Proceedings}, pages 54--60. {IEEE}
  Computer Society, 2004.

\bibitem[FK08]{FK08}
Alan Frieze and Ravi Kannan.
\newblock {A new approach to the planted clique problem}.
\newblock In {\em IARCS Annual Conference on Foundations of Software Technology
  and Theoretical Computer Science}. Schloss Dagstuhl-Leibniz-Zentrum f{\"u}r
  Informatik, 2008.

\bibitem[FKP19]{FKP19}
Noah Fleming, Pravesh Kothari, and Toniann Pitassi.
\newblock {Semialgebraic Proofs and Efficient Algorithm Design}.
\newblock {\em Foundations and Trends{\textregistered} in Theoretical Computer
  Science}, 14(1-2):1--221, 2019.

\bibitem[GS12]{GS12}
Venkatesan Guruswami and Ali~Kemal Sinop.
\newblock {Faster {SDP} Hierarchy Solvers for Local Rounding Algorithms}.
\newblock In {\em 53rd Annual {IEEE} Symposium on Foundations of Computer
  Science, {FOCS} 2012, New Brunswick, NJ, USA, October 20-23, 2012}, pages
  197--206. {IEEE} Computer Society, 2012.

\bibitem[GW95]{GW95}
Michel~X. Goemans and David~P. Williamson.
\newblock {Improved approximation algorithms for maximum cut and satisfiability
  problems using semidefinite programming}.
\newblock {\em Journal of the ACM (JACM)}, 42(6):1115--1145, 1995.

\bibitem[H{\aa}s01]{Has01}
Johan H{\aa}stad.
\newblock Some optimal inapproximability results.
\newblock {\em Journal of the ACM (JACM)}, 48(4):798--859, 2001.

\bibitem[HLZ10]{HLZ10}
Simai He, Zhening Li, and Shuzhong Zhang.
\newblock {Approximation algorithms for homogeneous polynomial optimization
  with quadratic constraints}.
\newblock {\em Mathematical Programming}, 125:353--383, 2010.

\bibitem[HV04]{HV04}
Johan H{\aa}stad and Srinivasan Venkatesh.
\newblock {On the advantage over a random assignment}.
\newblock {\em Random Structures \& Algorithms}, 25(2):117--149, 2004.

\bibitem[Jof74]{Jof74}
Anatole Joffe.
\newblock {On a Set of Almost Deterministic $k$-Independent Random Variables}.
\newblock {\em Ann. Probab.}, 2(6):161--162, 1974.

\bibitem[KN08]{KN08}
Subhash Khot and Assaf Naor.
\newblock {Linear Equations Modulo 2 and the $L_1$ Diameter of Convex Bodies}.
\newblock {\em SIAM Journal on Computing}, 38(4):1448, 2008.

\bibitem[Meg01]{Meg01}
Alexandre Megretski.
\newblock {Relaxations of Quadratic Programs in Operator Theory and System
  Analysis}.
\newblock In {\em Systems, Approximation, Singular Integral Operators, and
  Related Topics: International Workshop on Operator Theory and Applications,
  IWOTA 2000}, pages 365--392. Springer, 2001.

\bibitem[{Mon}90]{Mon90}
Stephen {Montgomery-Smith}.
\newblock {The Distribution of Rademacher Sums}.
\newblock {\em Proceedings of the American Mathematical Society},
  109(2):517--522, 1990.

\bibitem[MZ10]{MZ10}
Raghu Meka and David Zuckerman.
\newblock {Pseudorandom Generators for Polynomial Threshold Functions}.
\newblock In Leonard~J. Schulman, editor, {\em Proceedings of the 42nd {ACM}
  Symposium on Theory of Computing, {STOC} 2010, Cambridge, Massachusetts, USA,
  5-8 June 2010}, pages 427--436. {ACM}, 2010.

\bibitem[O'D14]{O14}
Ryan O'Donnell.
\newblock {\em Analysis of Boolean Functions}.
\newblock Cambridge University Press, 2014.

\bibitem[Par00]{Par00}
Pablo~A. Parrilo.
\newblock {\em Structured semidefinite programs and semialgebraic geometry
  methods in robustness and optimization}.
\newblock PhD thesis, California Institute of Technology, 2000.

\bibitem[Rag08]{Rag08}
Prasad Raghavendra.
\newblock {Optimal algorithms and inapproximability results for every CSP?}
\newblock In {\em Proceedings of the 40th Annual {ACM} Symposium on Theory of
  Computing, Victoria, British Columbia, Canada, May 17-20, 2008}, pages
  245--254, 2008.

\bibitem[SK75]{SK75}
David Sherrington and Scott Kirkpatrick.
\newblock Solvable model of a spin-glass.
\newblock {\em Physical review letters}, 35(26):1792, 1975.

\bibitem[ST21]{ST21}
David Steurer and Stefan Tiegel.
\newblock {SoS Degree Reduction with Applications to Clustering and Robust
  Moment Estimation}.
\newblock In D{\'{a}}niel Marx, editor, {\em Proceedings of the 2021 {ACM-SIAM}
  Symposium on Discrete Algorithms, {SODA} 2021, Virtual Conference, January 10
  - 13, 2021}, pages 374--393. {SIAM}, 2021.

\bibitem[Tre]{Trevisan}
Luca Trevisan.
\newblock {The Khot-Naor Approximation Algorithm for 3-XOR}.
\newblock Available at
  \url{https://lucatrevisan.wordpress.com/2021/10/12/the-khot-naor-approximation-algorithm-for-3-xor/}.
\newblock Accessed Feb. 10, 2023.

\end{thebibliography}

\appendix
\section{A simple \texorpdfstring{$O(\sqrt{n})$}{O(sqrt(n))}-certifiable upper bound}
\label{sec:warmup}

In this section, we provide a simple alternative certification algorithm achieving the same approximation ratio as \Cref{thm:deg6-sos-approx}.
Let $f(x,y,z) = \sum_{1\le i,j,k\le n} T_{ijk} x_i y_j z_k$ where $(T_{i,j,k})_{1\le i,j,k\le n}$ is a symmetric 3-tensor. We want to approximate
\begin{equation} \label{prob.def} 
    \max_{x,y,z \in \{\pm 1 \}^n} f(x,y,z) \mper
\end{equation}
Let $\calD$ be a pairwise independent distribution over $\pmo^n$. We can use a construction in which $|\supp(\calD)| = O(n)$. We will assume without loss of generality that if $\wh x\in \supp(\calD)$ then also $-\wh x \in \supp(\calD)$.

We consider the following approximation algorithm: for each $\wh x \in \supp(\calD)$, find a constant factor approximation of
\begin{equation} 
    \label{fix.x}  \max_{y,z \in \pmo^n} f(\wh{x},y,z) = \max_{y,z\in \pmo^n} \sum_{1\le i,j,k\le n} T_{i,j,k} \wh x_i y_j z_k\mcom
\end{equation}
using (the proof of)  Grothendieck's inequality. Output the best solution over all choices of $\wh x \in \supp(\calD)$.

Call $r$ the maximum of the Grothendieck relaxation of \Cref{fix.x} over all $\wh x\in \supp(\calD)$. The algorithm outputs a solution of value $\Omega(r)$. We want to prove that the standard degree-4 SoS relaxation of \Cref{prob.def} has optimum at most $r \cdot \sqrt n$, which establishes an $O(\sqrt n)$ integrality gap.

Let $\mu$ be the optimal pseudo-distribution for the degree-4 SoS relaxation of \Cref{prob.def} and $\SOS$ be its value.
Let $q_i(y,z) = y^\top T_i z = \sum_{1 \leq j,k \leq n} T_{i,j,k} y_j z_k$, and write $q = (q_1,\dots,q_n)$ such that $f(x,y,z) = \iprod{x,q}$.
We have that
\begin{align*}
    \SOS^2 = \Paren{\pE_{\mu} \iprod{x,q} }^2
    \leq \pE_{\mu}[\iprod{x,q}^2]
    \leq  n \cdot \pE_{\mu}\|q\|_2^2
\end{align*}
by Cauchy-Schwarz (\Cref{fact:sos-cauchy-schwarz}).

On the other hand, by definition of $r$ and our assumption that $\supp(\calD)$ is symmetric, for every $\wh x\in \supp(\calD)$ there is a degree-2 SoS proof (over variables $y$, $z$) that
\[
    f(\wh{x}, y,z) \leq r \text{ and } f(\wh{x},y,z) \geq - r \mper
\]
Hence, for every $\wh x \in \supp(\calD)$  there is a degree-4 SoS proof that
\[
    f(\wh{x},y,z)^2 = \iprod{\wh{x},q}^2 \leq r^2 \mcom
\]
which means that
\[
    r^2 \geq  \pE_\mu \E_{\wh x \sim \calD} \iprod{\wh{x},q}^2 = \pE_\mu \|q\|_2^2 \geq \frac{1}{n} \cdot \SOS^2 \mcom
\]
where we use the fact that the pairwise independence of $\calD$ implies that $\E_{\wh x \sim \calD} \iprod{\wh{x}, v}^2 = \|v\|_2^2$ for all $v \in \R^n$. Putting things together we get $\SOS^2 \leq n\cdot r^2$, which completes the argument.

\section{Optimizing higher-degree polynomials}
\label{sec:high-degree}

\subsection{Decoupling inequalities}

We first prove a decoupling lemma for all odd-degree polynomials.

\begin{lemma}[High-degree version of \Cref{lem:decoupling}]
    \label{lem:decoupling-hd}
    Let $d$ be an odd integer.
    Let $f(x) = \iprod{T, x^{\otimes d}}$ be a multilinear homogeneous degree-$d$ polynomial in $n$ variables (where $T$ is a symmetric $d$-tensor),
    and let $\wt{f}(x^{(1)}, \dots,x^{(d)}) = \iprod{T, x^{(1)} \otimes \cdots \otimes x^{(d)} }$ be the decoupled polynomial of $f$.
    Then, given any $x^{(1)}, \dots, x^{(d)} \in \pmo^n$, there exists a sampleable distribution $\calD$ over $\pmo^n$ such that
    \begin{equation*}
        \E_{y\sim \calD}\left[f(y)\right] = \frac{d!}{d^d} \cdot \wt{f}\Paren{ x^{(1)}, \dots, x^{(d)} } \mper
    \end{equation*}
    As a consequence,
    \begin{equation*}
        \max_{y\in \pmo^n} f(y) \geq \frac{d!}{d^d} \cdot \max_{x^{(1)}, \dots, x^{(d)} \in \pmo^n}  \wt{f}\Paren{ x^{(1)}, \dots, x^{(d)} } \mper
    \end{equation*}
\end{lemma}
\begin{proof}
    The distribution $\calD$ can be sampled as follows,
    \begin{itemize}
        \item Let $b_1,\ldots,b_{d-1}$ be i.i.d. uniform $\pm 1$ random variables and let $b_d\coloneqq b_1\ldots b_{d-1}$. Note that the distribution of $b=(b_1,\ldots,b_d)$ is $(d-1)$-wise independent and $b_1 b_2 \cdots b_d = 1$.
        \item Independently for each $i \in [n]$, sample $y_i$ uniformly from $\{b_j x_i^{(j)}\}_{j\in[d]}$.
    \end{itemize}
    Since each $y_i$ is sampled independently conditioned on $b$, we have that for any pairwise distinct indices $i_1,\dots,i_d \in [n]$,
    \begin{equation*}
        \E_{y\sim \calD}\left[y_{i_1}\cdots y_{i_d}\right] = \E_b \left[\prod_{k=1}^d \Paren{\frac{1}{d} \sum_{j=1}^d b_j x_{i_k}^{(j)} } \right]
        = d^{-d} \sum_{j_1,\dots,j_d \in [d]} \E_b\left[b_{j_1}\cdots b_{j_d}\right] \cdot x_{i_1}^{(j_1)} \cdots x_{i_d}^{(j_d)}  \mper
    \end{equation*}
    Since $b$ follows a $(d-1)$-wise independent distribution, $b_1 b_2 \cdots b_d = 1$ and $d$ is odd, $\E_b\left[b_{j_1}\cdots b_{j_d}\right]$ does not vanish (and equals 1) if and only if $j_1,\dots,j_d$ are distinct, i.e., $\{j_1,\dots,j_d\} = [d]$.
    Thus, the summation above is simply summing over permutations of $[d]$:
    \begin{equation*}
        \E_{y\sim \calD}\left[y_{i_1}\cdots y_{i_d}\right] = \frac{d!}{d^d} \cdot \E_{\pi\sim \bbS_d}\Brac{ x_{i_1}^{(\pi(1))} \cdots x_{i_d}^{(\pi(d))} }
        = \frac{d!}{d^d} \cdot \E_{\pi\sim \bbS_d}\Brac{ x_{i_{\pi(1)}}^{(1)} \cdots x_{i_{\pi(d)}}^{(d)} }
        \mcom
    \end{equation*}
    where $\pi \sim \bbS_d$ denotes a random permutation of $[d]$. Finally, as $T$ is a symmetric tensor, we deduce
    \begin{equation*}
        \E_{y\sim \calD} \left[f(y)\right] = \frac{d!}{d^d} \cdot \wt{f}\Paren{ x^{(1)}, \dots, x^{(d)} } \mper
    \end{equation*}
    This proves the first statement of the lemma.
    The second statement follows immediately.
\end{proof}

For even-degree polynomials, \Cref{lem:decoupling-hd} simply cannot hold.
This is easy to appreciate by discussing the simple setting of quadratics:
\begin{example}[Impossibility of decoupling for quadratics]
\label{ex:failure-of-decouling}
    Consider the matrix $Q = I - \vec{1}\vec{1}^\top$, and define the multilinear quadratic polynomial $f(x) = x^\top Qx$ and the decoupled polynomial $\wt{f}(x,y) = x^\top Q y$.
    Then, it is easy to verify that $\max_{x\in\pmo^n} f(x) = n$ but $\max_{x,y\in\pmo^n} \wt{f}(x,y) = n^2-n$, which is a $\poly(n)$ gap.
\end{example}

On the other hand, if we only consider $\max_{y} |f(y)|$ like in the setting of \cite{BGG+17} (as opposed to $\max_y f(y)$), then decoupling inequalities with the same guarantees as \Cref{lem:decoupling-hd} hold for \emph{any} degree:

\begin{lemma}[Decoupling for absolute values, any degree]
    \label{lem:decoupling-abs-2}
    Let $d \in \N$.
    Let $f(x) = \iprod{T, x^{\otimes d}}$ be a multilinear homogeneous degree-$d$ polynomial in $n$ variables (where $T$ is a symmetric $d$-tensor),
    and let $\wt{f}(x^{(1)}, \dots,x^{(d)}) = \iprod{T, x^{(1)} \otimes \cdots \otimes x^{(d)} }$ be the decoupled polynomial of $f$.
    Then,
    \begin{equation*}
        \max_{y\in \pmo^n} |f(y)| \geq \frac{d!}{d^d} \cdot \max_{x^{(1)}, \dots, x^{(d)} \in \pmo^n}  \wt{f}\Paren{ x^{(1)}, \dots, x^{(d)} } \mper
    \end{equation*}
\end{lemma}
\begin{proof}
    We use the trick that $\max_{y\sim \pmo^n}|f(y)| = \max_{y\in [-1,1]^n} |f(y)|$.
    Thus, given assignments $x^{(1)},\dots,x^{(d)} \in \pmo^n$, it suffices to round to a $y\in [-1,1]^n$.
    
    We next state the well-known polarization identity for degree-$d$ homogeneous polynomials:
    \begin{equation*}
        \wt{f}\Paren{x^{(1)}, \dots, x^{(d)}} = \E_{\eps\sim \pmo^n} \Brac{\frac{\eps_1 \eps_2 \cdots \eps_d}{d!} f\Paren{\eps_1 x^{(1)} + \cdots + \eps_d x^{(d)}}} \mper
    \end{equation*}
    Define $y_{\eps} \seteq \frac{1}{d}(\eps_1 x^{(1)} + \cdots \eps_d x^{(d)}) \in [-1,1]^n$.
    Then, rewriting the above and using the triangle inequality,
    \begin{equation*}
        \wt{f}\Paren{x^{(1)}, \dots, x^{(d)}} = \frac{d^d}{d!} \cdot \E_{\eps\sim\pmo^n}\Brac{\eps_1\eps_2 \cdots \eps_d \cdot f(y_{\eps})}
        \leq \frac{d^d}{d!} \cdot \E_{\eps\sim\pmo^n} \Brac{ |f(y_{\eps})| } \mper
    \end{equation*}
    Thus, there exists a $y_{\eps}\in [-1,1]^n$ such that $|f(y_{\eps})| \geq \frac{d!}{d^d}\cdot \wt{f}(x^{(1)},\dots, x^{(d)})$.
\end{proof}

\subsection{Rounding SoS relaxations for high-degree polynomials}

We now give a simple polynomial-time certification and rounding algorithm using the canonical sum-of-squares relaxation that achieves approximation $O(n^{\frac{d}{2}-1})$ for optimizing \emph{decoupled} homogeneous degree-$d$ polynomials over the hypercube.
The proof is essentially identical to the cubic case (\Cref{thm:deg6-sos-approx}).

\begin{theorem} \label{thm:hd-poly}
    Let $d\ge 3$ and $n \ge 1$ be integers.
    Given any decoupled homogeneous degree-$d$ polynomial $f(x^{(1)}, \dots, x^{(d)}) = \iprod{T, x^{(1)} \otimes \cdots \otimes x^{(d)} }$, the degree-$2d$ SoS relaxation of
    \[
    \max_{x^{(1)}, \dots, x^{(d)}\in\pmo^n} f(x^{(1)}, \dots, x^{(d)})
    \]
    has integrality gap at most $O(n^{\frac{d}{2}-1})$.
    Furthermore, given a degree-$2d$ pseudo-distribution $\mu$ such that $\SOS \coloneqq \pE_{\mu}f>0$,
     there is a randomized $n^{O(d)}$-time rounding algorithm that outputs $\ol{x}^{(1)},\dots,\ol{x}^{(d)} \in \pmo^n$ such that with high probability $f(\ol{x}^{(1)}, \dots, \ol{x}^{(d)}) \geq \Omega(n^{-\frac{d}{2}+1}) \cdot \SOS$.
\end{theorem}

\begin{proof}
    For $i_3,\dots,i_d\in [n]$, we let  $q_{i_3,\dots,i_d}(x^{(1)}, x^{(2)}) \coloneqq \iprod{T_{i_3,\dots,i_d}, x^{(1)} \otimes x^{(2)}}$, a degree-2 polynomial in $x^{(1)}$ and $x^{(2)}$, where $T_{i_3,\dots,i_d}$ is an $n\times n$ matrix corresponding to a slice of the tensor $T$.
    For simplicity of notation, we will drop the dependence on $x^{(1)}$ and $x^{(2)}$ and write $Q \coloneqq Q(x^{(1)},x^{(2)})= (q_{i_3,\dots,i_d})_{i_3,\dots,i_d\in[n]}$ as an order-$(d-2)$ tensor whose entries are degree-2 polynomials  in $x^{(1)}$ and $x^{(2)}$.
    Then, we have
    \begin{align*}
        \SOS &= \sum_{i_3,\dots,i_d \in [n]} \pE_{\mu} \Brac{q_{i_3,\dots,i_d} x_{i_3}^{(3)} \cdots x_{i_d}^{(d)} } \\
        &\leq \sum_{i_3,\dots,i_d \in [n]} \sqrt{\pE_{\mu} \Brac{q_{i_3,\dots,i_d}^2} } \\
        &\leq \sqrt{n^{d-2} \sum_{i_3,\dots,i_d \in [n]} \pE_{\mu} \Brac{q_{i_3,\dots,i_d}^2} }
        = n^{\frac{d}{2}-1}\sqrt{\pE_{\mu}\, \|Q\|_F^2 }
        \mcom
    \end{align*}
    using Cauchy-Schwarz and its pseudo-expectation version.
    Here, $\|Q\|_F^2$ is the sum of the squared coefficients of the tensor $Q$.
    However, taking $h^{(3)},\dots, h^{(d)}$ to be i.i.d.\ uniform samples from $\pmo^n$, we have
    \begin{align*}
        \E_{h^{(3)},\dots,h^{(d)}} \iprod{Q, h^{(3)}\otimes \cdots \otimes h^{(d)} }^2
        = \|Q\|_F^2 \mper
    \end{align*}
    For simplicity, we denote $H \coloneqq h^{(3)} \otimes \cdots \otimes h^{(d)}$. Then,
    \begin{equation*}
        \SOS^2 \leq n^{d-2} \cdot \E_{H}\, \pE_{\mu}\, \Iprod{Q, H}^2 \mcom
    \end{equation*}
    where we recall that the coefficients of $Q$ are degree-2 polynomials in $x^{(1)}$ and $x^{(2)}$.
    We now describe the rounding algorithm.
    \begin{enumerate}
        \item Sample $h^{(3)},\dots, h^{(d)} \sim \pmo^n$, and set $\ol x^{(j)} \coloneqq h^{(j)}$ for $j \geq 3$.
        Denote $H \coloneqq h^{(3)} \otimes \cdots \otimes h^{(d)}$.

        \item Apply \Cref{lem:scalar-fixing} to obtain from $\mu$ a degree-$2$ pseudo-distribution $\mu'$ satisfying 
        \[
            \left|\pE_{\mu'} \iprod{Q, H}\right| \geq \frac{1}{3} \sqrt{\pE_{\mu} \iprod{Q,H}^2}\mper
        \]

        \item Use Grothendieck rounding (\Cref{fact:grothendieck}) on $\mu'$ to obtain solutions $\ol x^{(1)}, \ol x^{(2)} \in \pmo^n$ such that
        $\iprod{Q(x^{(1)}, x^{(2)}), H}\ge \frac{1}{K_G}\cdot \left|\pE_{\mu'} \iprod{Q, H}\right|$ (we can always flip the sign of $h^{(3)}$ to get the guarantee with the absolute value).
    \end{enumerate}
    First, note that $\pE_{\mu} \iprod{Q,H}^2$ is a degree-$2(d-2)$ polynomial in $h^{(3)},\ldots,h^{(d)}$, thus by \Cref{lem:anti-concentration-hypercontractivity} and hypercontractivity of polynomials over the hypercube, we know that
    \begin{equation*}
        \Pr_{H} \Brac{ \pE_{\mu}\,\iprod{Q,H}^2 \geq \E_H\, \pE_{\mu}\,\iprod{Q,H}^2 } \geq 2^{-O(d)} \mper
    \end{equation*}
    Hence, with probability at least $2^{-O(d)}$, we get a ``good'' $H$ such that $\pE_{\mu}\iprod{Q,H}^2 \geq n^{-(d-2)} \cdot \SOS^2$. Putting everything together, with probability at least $2^{-O(d)}$ we obtain $\ol x^{(1)},\ldots,\ol x^{(d)}$ such that $f(\ol x^{(1)},\ldots,\ol x^{(d)})\ge \Omega(n^{-\frac{d}{2}+1})\cdot \SOS$.
    Repeating the above $\poly(n,2^{d})$ times, we can obtain a solution with value $\Omega(n^{-\frac{d}{2}+1}) \cdot \SOS$ with high probability.
    This completes the proof.
\end{proof}

By combining \Cref{thm:hd-poly} with our decoupling inequalities \Cref{lem:decoupling-hd} and \Cref{lem:decoupling-abs-2}, we deduce that the same approximation guarantees hold in the general, non-decoupled case, for maximizing an odd-degree homogeneous polynomial, or maximizing the absolute value of an homogeneous polynomial of any degree. While we stated our results on the hypercube, the same holds for maximizing over the unit sphere, using the Gaussian rounding from \Cref{thm:sos-spherical} instead of Grothendieck's inequality.

\end{document}